\documentclass[10pt,journal,final,twocolumn]{IEEEtran}
\IEEEoverridecommandlockouts

\usepackage{amsmath,amssymb,amsfonts}
\usepackage{cite}
\usepackage{graphicx}
\usepackage{epstopdf}

\usepackage{booktabs}
\usepackage{graphicx}
\usepackage{textcomp}
\usepackage{xcolor}
\usepackage{times}
\usepackage{subfigure}
\usepackage{amssymb,amsmath}
\usepackage{acronym}  
\usepackage{balance}

\usepackage{bm}         
\usepackage{algorithm} 
\usepackage{algorithmic}

\usepackage{lipsum}
\usepackage{stfloats}

\usepackage{url}
\usepackage{mathrsfs}
\usepackage{bbm}

\usepackage{amsthm}
\usepackage{enumerate}

\usepackage{hyperref}
\hypersetup{hidelinks, 
colorlinks=true,
allcolors=black,
pdfstartview=Fit,
breaklinks=true}

\theoremstyle{definition}

\newtheorem{lemma}{\bf Lemma}
\newtheorem{corollary}{\bf Corollary}

\theoremstyle{remark}
\newtheorem{remark}{\bf Remark}

\acrodef{ofdm}[OFDM]{orthogonal frequency division multiplexing}%
\acrodef{miso-ofdm}[MISO-OFDM]{multi-input single-output orthogonal frequency division multiplexing}%

\acrodef{ris}[RIS]{reconfigurable intelligent surface}%
\acrodef{qos}[QoS]{quality of service}%
\acrodef{idft}[IDFT]{inverse discrete Fourier transform}%
\acrodef{dft}[DFT]{discrete Fourier transform}%
\acrodef{cp}[CP]{cyclic prefix}%
\acrodef{csi}[CSI]{channel state information}%
\acrodef{awgn}[AWGN]{additive white Gaussion noise}%

\acrodef{qcqp}[QCQP]{quadratically constrained quadratic program}%
\acrodef{qp}[QP]{quadratic program}%
\acrodef{bs}[BS]{base station}%
\acrodef{ap}[BS]{base station}%
\acrodef{aps}[APs]{access points}%
\acrodef{qos}[QoS]{quality of service}%
\acrodef{ue}[UE]{user equipment}%
\acrodef{snr}[SNR]{signal-to-noise ratio}%
\acrodef{mmwave}[mmWave]{millimeter-wave}%
\acrodef{snr}[SNR]{signal-to-noise ratio}%

\acrodef{sinr}[SINR]{signal-to-interference-plus-noise ratio}%

\acrodef{ser}[SER]{symbol error rate}%
\acrodef{rc}[RC]{reflection coefficient}%
\acrodef{uavs}[UAVs]{unmanned aerial vehicles}%
\acrodef{mimo}[MIMO]{multiple-input multiple-output}%
\acrodef{noma}[NOMA]{non-orthogonal multiple access}%

\acrodef{ace}[ACE]{adaptive cross-entropy}%
\acrodef{wsr}[WSR]{weighted sum-rate}%
\acrodef{udn}[UDN]{ultra-dense network}%
\acrodef{Udn}[UDN]{Ultra-dense network}%

\def\BibTeX{{\rm B\kern-.05em{\sc i\kern-.025em b}\kern-.08em
    T\kern-.1667em\lower.7ex\hbox{E}\kern-.125emX}}

\begin{document}
\title{Multiple Access for Near-Field Communications: SDMA or LDMA?}
\author{\IEEEauthorblockN{Zidong Wu and Linglong Dai, \emph{Fellow, IEEE}}

\thanks{Manuscript received 24 October 2022; revised 18 February 2023; accepted 10 April 2023. Date of publication 12 May 2023; date of current version 19 June 2023. This work was supported in part by the National Key Research and Development Program of China (Grant No. 2020YFB1807201) and in part by the National Natural Science Foundation of China (Grant No. 62031019). An earlier version of this paper was presented in part at the IEEE ICC, Rome, Italy, in May 2023. (Corresponding author: {\emph{Linglong Dai}}.)}
\thanks{The authors are with the Department of Electronic Engineering, Tsinghua University, Beijing 100084, China, and also with the Beijing National Research Center for Information Science and Technology (BNRist), Beijing 100084, China. (e-mails: wuzd19@mails.tsinghua.edu.cn, daill@tsinghua.edu.cn).}
\thanks{Color versions of one or more figures in this article are available at
https://doi.org/10.1109/JSAC.2023.3275616.}
\thanks{Digital Object Identifier 10.1109/JSAC.2023.3275616}
}
\maketitle
\begin{abstract}
Spatial division multiple access (SDMA) is essential to improve the spectrum efficiency for multi-user multiple-input multiple-output (MIMO) communications. The classical SDMA for massive MIMO with hybrid precoding heavily relies on the angular orthogonality in the far field to distinguish multiple users at different angles, which fails to fully exploit spatial resources in the distance domain. With the dramatically increasing number of antennas, the extremely large-scale antenna array (ELAA) introduces additional resolution in the distance domain in the near field. In this paper, we propose the concept of location division multiple access (LDMA) to provide a new possibility to enhance spectrum efficiency compared with classical SDMA. The key idea is to exploit extra spatial resources in the distance domain to serve different users at different locations (determined by angles and distances) in the near field. Specifically, the asymptotic orthogonality of near-field beam focusing vectors in the distance domain is proved, which reveals that near-field beam focusing is able to focus signals on specific locations with limited leakage energy at other locations. This special property could be leveraged in hybrid precoding to mitigate inter-user interferences for spectrum efficiency enhancement. Moreover, we provide the spherical-domain codebook design method for LDMA communications with the uniform planar array, which provides the sampling method in the distance domain. Additionally, performance analysis of LDMA is provided to reveal that the asymptotic optimal spectrum efficiency could be achieved with the increasing number of antennas. Finally, simulation results verify the superiority of the proposed LDMA over SDMA in different scenarios.
\end{abstract}


\section{Introduction}
\par To meet the dramatically increasing demand for data transmission, spectrum efficiency has always been considered as one of the most important key performance indicators (KPIs) for communications \cite{Navrati'16'}. In the fifth generation (5G) networks, by employing dozens or hundreds of antennas at the base station (BS) to simultaneously serve multiple user equipments (UEs), massive multiple-input multiple-output (MIMO) is the key enabler for achieving the remarkable growth in spectrum efficiency \cite{Zhangrui'14'j,Marzetta'14'j,Jackb'20'j}. In 5G massive MIMO communications, spatial division multiple access (SDMA) is indispensable to support multi-stream transmissions with multiple UEs. Through SDMA, different UEs are able to establish parallel transmission links with BS by exploiting different spatial resources, which significantly boosts spectrum efficiency.
\par With the requirement of a 10-fold increase in spectrum efficiency for future 6G communications, massive MIMO technology will evolve into extremely large-scale antenna array (ELAA) equipped with thousands of antennas. Thanks to the significantly increased number of antennas, ELAA is capable of further improving the spectrum efficiency by orders of magnitude for 6G~\cite{Fredrik'21'j}. The research on ELAA has attracted significant attention from academia and industry, and the implementations of ELAA have been validated in over-the-air channel measurement platforms~\cite{wangchengxiang'22'vtc,Karstensen'22'tap}.
\subsection{Prior Works}\label{sec:intro prior}
\par From MIMO for 4G, massive MIMO for 5G, and ELAA for future 6G, people are seeking the improvement of spectrum efficiency mainly by increasing the number of antennas to support more parallel data transmissions under the SDMA framework. For SDMA, precoding plays a central role in forming high-gain beams and managing inter-user interferences to enhance spectrum efficiency. There are plenty of fully digital precoding techniques to realize SDMA in MIMO systems~\cite{Marzetta'10,Tufvesson'13,Debbah'13,Vincent'19,Rao'17}. Specifically, in~\cite{Marzetta'10}, a matched-filter (MF) precoding method was proved to be able to achieve optimal spectrum efficiency with independently and identically distributed channels for an infinite number of UEs. Then, it was shown in~\cite{Tufvesson'13} that the zero-forcing (ZF) precoding could achieve higher spectrum efficiency than MF precoding when the number of transmit antennas is much larger than the number of UEs. To avoid the ill-posed inversion process brought by the ZF precoding, regularized zero-forcing (RZF) precoding scheme was introduced to improve the system performance when the number of transmit antennas is not far beyond the served UEs~\cite{Debbah'13,Vincent'19}.

\par With the increasing number of antennas, the power consumption and hardware cost of fully-digital precoding become overwhelming for massive MIMO and ELAA systems~\cite{Gao'16'j,Wangzhao'16'j}. To significantly reduce power consumption, the hybrid precoding architecture has been widely considered~\cite{Ayach'14'j}, where energy-efficient phase shifters are employed to realize analog precoding to reduce the number of radio frequency (RF) chains compared with fully-digital precoding. In hybrid precoding architectures, SDMA is usually realized by the joint design of analog and digital precoders. Since the analog precoding is realized by phase shifters, the constant modulus constraints of phase shifters impose difficulties on the analog precoder designs, which has become the main challenge for SDMA.

\par The analog precoder design methods to realize SDMA can be mainly classified into two categories: Optimization-based scheme and beam-steering-based scheme. In the first category, optimization tools are employed to overcome the constant modulus constraints to design analog precoders. Specifically, an alternate optimization of the analog precoder and the ZF-based digital precoder was proposed in~\cite{Yu'16'j2} to maximize the spectrum efficiency for SDMA. To further ensure the convergence, RZF-based digital precoding was considered in~\cite{Park'17'j2}, and an alternate optimization was performed to optimize the analog precoder and the compensation matrix, which is a component of the digital precoder. Nevertheless, those optimization-based schemes often involve high complexity, which is difficult to be applied to practical wireless systems.

\par To reduce the complexity, the other category of beam-steering methods usually employs the beam steering vectors to construct the analog precoder~\cite{Raghavan'17'j,Xiao'15'j,Heath'15'j}. By exploiting the sparsity of massive MIMO channels, particularly at high frequencies, the beam steering vectors corresponding to directional beams can be directly utilized to design analog precoders without complicated optimization tools, which could also approach the optimal spectrum efficiency in single-user scenarios~\cite{Ayach'14'j}. Moreover, the near-optimal single-user hybrid precoding design was generalized to the more general multi-user cases in~\cite{Raghavan'17'j}, where the authors pointed out that the directional beams can be also exploited to increase the spectrum efficiency in SDMA, although some inter-user interferences will be introduced.

\par To further eliminate the inter-user interferences, a practical two-stage multi-user precoding scheme was proposed in~\cite{Heath'15'j} to employ the steering beams to construct analog precoders. In the first stage, the analog precoder is selected from a predefined beam codebook to maximize the received signal power. Normally, the discrete Fourier transform (DFT) codebook is adopted to generate beams aligned with specific directions. In this stage, different UEs can be partially distinguished and served by different directional beams, and the inter-user interferences can be partially alleviated. Then, in the second stage, the digital precoder is designed to further eliminate the remained inter-user interferences. Moreover, owing to the asymptotic orthogonality of different directional beams in the angular domain, the beam steering vectors adopted are asymptotically orthogonal when the number of antennas tends to infinity\cite{Xiao'15'j}. That is to say, the received signal power could be maximized while inter-user interferences could be naturally eliminated in SDMA for a large number of antennas~\cite{Sohrabi'17'jsac}. Due to the low complexity and asymptotic angular orthogonality of beams, the beam-steering-based schemes are widely considered in 5G and are expected to be applied in future 6G.

\par Following such methods, to further enhance the spectrum efficiency, communication systems commonly rely on the consistent increase of the number of antennas to provide thinner beams, where analog precoding could eliminate inter-user interference more thoroughly. Although hybrid precoding is considered to be energy-efficient, the dramatic increase of antennas will bring unaffordable energy consumption. Besides increasing the number of antennas only, is there any other new possibility to significantly boost the spectrum efficiency?

\subsection{Our Contributions}\label{sec:intro contr}
Inspired by near-field communications recently investigated in~\cite{cui'22'm}, this paper tries to exploit the extra resolution in the distance domain brought by near-field beams to enhance the spectrum efficiency. Specifically, the transition from massive MIMO to ELAA implies that many communications happen in \emph{spherical-wave}-based near-field regions, instead of classical \emph{planar-wave}-based far-field regions. Owing to the different electromagnetic wave propagation models, unlike the far-field beams focusing signal energy on a certain angle, near-field beams are capable of focusing signal energy on a specific location~\cite{Heath'22'j}, which could be leveraged to mitigate inter-user interference and therefore improve the spectrum efficiency.
\par In this paper, the concept of location division multiple access (LDMA) is proposed, which provides a new possibility to exploit spatial resources to improve spectrum efficiency performance compared with classical SDMA\footnote{Simulation codes are provided to reproduce the results in this paper:
http://oa.ee.tsinghua.edu.cn/dailinglong/publications/publications.html.}. The contributions are summarized as follows:
\begin{itemize}
	\item The LDMA communication scheme is proposed in this paper. The key idea is to exploit the energy-focusing property of near-field beams to serve different UEs located at different angles and different distances, i.e. locations, to fully exploit the extra resources in the distance domain to improve the system performance. In classical SDMA, more antennas are required to increase the angular focusing property of beams to suppress the interference for spectrum efficiency improvement. On the contrary, benefiting from the extra focusing property of near-field beams in the distance domain~\cite{Cui'22'tcom,Yonina'22'm}, UEs located not only at different angles but also at different distances can also be served by near-field beams without serious interference. The proposed LDMA scheme poses another dimension to improve the spectrum efficiency compared with classical SDMA.

	\item Under the LDMA framework, similar to the asymptotic orthogonality of far-field beams in the angular domain~\cite{Xiao'15'j}, the asymptotic orthogonality of near-field beams in the extra distance domain is investigated. For different array structures, we prove the property that as the number of antennas tends to infinity, the correlation of near-field beams focusing on the same angle but different distances converges to zero based on Fresnel approximation. The extra asymptotic orthogonality of near-field beams in the distance domain ensures the theoretical feasibility of LDMA with large antenna arrays.

	\item To efficiently utilize the extra orthogonality of near-field beams in the distance domain proved above, a three-dimensional (3D) near-field codebook for uniform planar array (UPA) systems is designed to distinguish different UEs. The classical far-field DFT codebook only performs sampling in different azimuth and elevation angles. In addition to the uniform angular sampling, the near-field codebook could characterize the correlation of near-field beams in the distance domain to perform efficient distance sampling. Different from the uniform angular sampling, we provide a non-uniform sampling criterion in the distance domain for UPA systems.

	\item Finally, by virtue of the asymptotic orthogonality of near-field beams both in angular and distance domains, the performance analysis is provided to reveal the asymptotic optimality of the LDMA scheme. That is to say, when the number of antennas tends to infinity, the spectrum efficiency of LDMA could approach the interference-free scenarios. Simulation results are also provided to verify the superiority of the LDMA scheme compared with classical SDMA by providing a new possibility to enhance the spectrum efficiency.

\end{itemize}
\subsection{Organization and Notation}\label{sec:intro org}
The remainder of the paper is organized as follows. Section~\ref{sec: sys} introduces the system model and near-field beam focusing vectors for uniform
linear array (ULA) and UPA systems. Section~\ref{sec: analysis} theoretically investigates the asymptotic orthogonality of near-field beams. In Section~\ref{sec: LDMA}, the LDMA scheme is illustrated and the spherical-domain codebook design method is provided. The performance analysis is detailed in Section~\ref{sec: per ana}. Simulation results are provided in Section~\ref{sec: sim}, and conclusions are drawn in Section~\ref{sec: conclusion}.

\textit{Notations}: $\mathbb{C}$ denotes the set of complex numbers; ${[\cdot]^{-1}}$, ${[\cdot]^{T}}$, ${[\cdot]^{H}}$ and ${\rm diag}(\cdot)$ denote the inverse, transpose, conjugate-transpose and diagonal operations, respectively; $\|\cdot\|_F $ denotes the Frobenius norm of a matrix; $|\cdot|$ denotes the norm of its complex argument; $\lesssim$ denotes approximately less or equal to; $\mathbf{I}_{N}$ is an $N\times N$ identity matrix.

\section{System Model}\label{sec: sys}
In this section, we introduce the hybrid precoding architecture for ELAA systems. The classical far-field and near-field channel models are also investigated and compared.
\subsection{System Model}\label{sec: sys sys}
We consider a narrowband time division duplexing (TDD) ELAA single-cell communication scenario in this paper. The BS is equipped with an $N$-antenna ELAA and $N_{\rm{RF}}$ RF chains, where the hybrid precoding architecture is employed and $N_{\rm{RF}} \leq N$ is satisfied. The BS aims to simultaneously serve $K$ single-antenna UEs in the cell, which requires $N_{\rm{RF}} \geq K$. For analysis simplicity, $N_{\rm{RF}} = K$ is assumed throughout the paper. To enable the multiple access for different UEs, the uplink and downlink transmissions are both considered. The downlink system model is first introduced and the uplink channel can be obtained by transposing the downlink model according to the channel reciprocity~\cite{Xiao'15'j}.

\begin{figure}[!t]
	\centering
	\setlength{\abovecaptionskip}{0.cm}
	\includegraphics[width=3in]{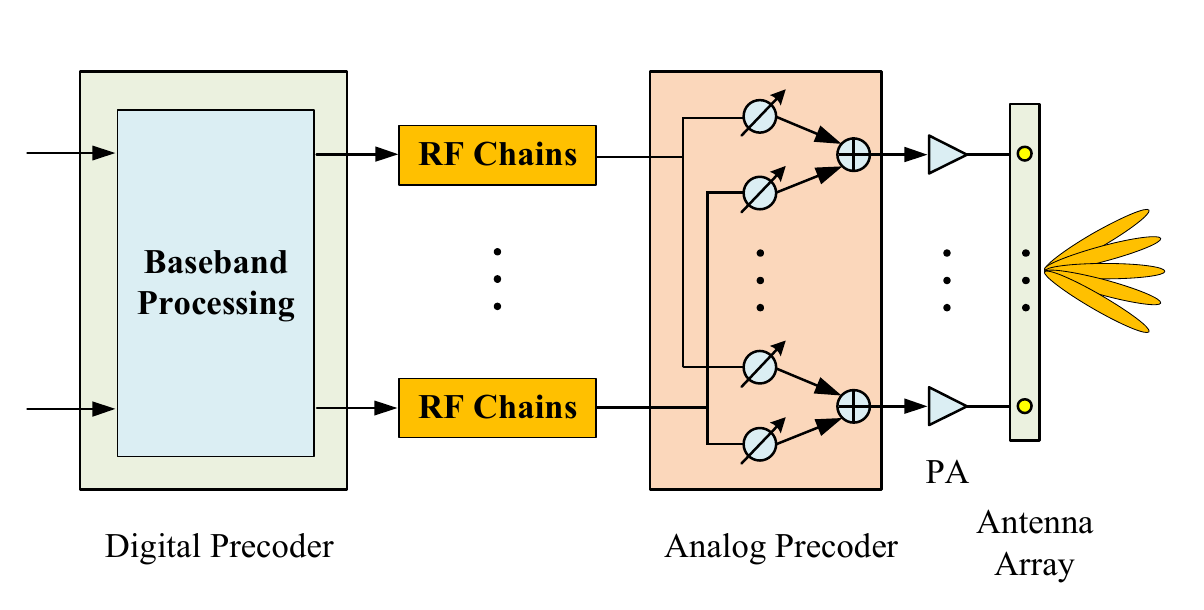}
	\caption{Classical hybrid precoding architecture of massive MIMO communication systems.}
	\label{img: sys model}
\end{figure}

A traditional massive MIMO millimeter-wave (mmWave) communication system can be modeled as in Fig.~\ref{img: sys model}. The received signal for all $K$ UEs can be represented as
\begin{equation}
\label{eq: downlink received signal}
\begin{aligned}
{\bf{y}}^{\rm{DL}} = {\bf{H}} {\bf{F}}_{\rm{A}} {\bf{F}}_{\rm{D}} {\bf{s}} + {\bf{n}},
\end{aligned}
\end{equation}
where ${\bf{y}}^{\rm{DL}} = [y_1, y_2, \cdots, y_K]^T$ denotes the $K \times 1$ received signals for all UEs, ${\bf{H}} = [{\bf{h}}_1, {\bf{h}}_2, \cdots, {\bf{h}}_K]^H$ denotes the downlink channel matrix, and ${\bf{h}}_k$ denotes the channel vector between the base station and the $k^{th}$ UE. The signal vector ${\bf{s}}$ satisfying the power constraint $\mathbb{E}[{\bf{s}}{\bf{s}}^H] = {\bf{I}}$ is transmitted to all UEs. The precoding matrices contain two separate components, i.e., digital precoder ${\bf{F}}_{\rm{D}}$ and analog precoder ${\bf{F}}_{\rm{A}}$. Finally, ${\bf{n}} \sim \mathcal{CN}(0, \sigma_n^2{\bf{I}}_{K})$ is the additive white Gaussian noise (AWGN), where $\sigma_n^2$ denotes the variance of the noise.  
\par The overall spectrum efficiency could be expressed as
\begin{equation}
\label{eq: spectrum efficiency}
\begin{aligned}
R =  \sum_k R_k = \sum_k \log_2\left(1+\frac{p_k|{\bf{h}}_k^H{\bf{F}}_{\rm{A}}{\bf{f}}_{{\rm{D}},k}|^2}{\sigma_n^2 + \sum_{l \neq k} p_l|{\bf{h}}_k^H{\bf{F}}_{\rm{A}}{\bf{f}}_{{\rm{D}},l}|^2}\right),
\end{aligned}
\end{equation}
where $p_k$ denotes the power allocated to the $k^{th}$ UE, ${\bf{f}}_{{\rm{D}},k}$ denotes the $k^{th}$ column of digital precoder ${\bf{F}}_{\rm{D}}$.

\subsection{Far-Field Channel Model}\label{sec: sys channel far}
The wireless channel can be constructed by either the far-field model~\cite{Ayach'14'j} or the near-field model~\cite{Cui'22'tcom}. Rayleigh distance is commonly adopted to identify the boundary between far-field and near-field regions, which is defined as $r_{\rm{RD}} = \frac{2D^2}{\lambda}$~\cite{Sherman'62'j}\footnote{Apart from Rayleigh distance, recently several criteria have been proposed to partition the far-field and near-field region more accurately, including the effective Rayleigh distance~\cite{cui'21} and Bj{\"o}rnson distance~\cite{emil'22}. To ensure the conciseness of expressions, we adopt the classical Rayleigh distance to characterize the near-field region in this paper.}. In 5G massive MIMO communications, since the aperture of antenna array is not very large, the Rayleigh distance is very limited. For instance, the Rayleigh distance of a 32-element array at 30 GHz is 5 meters, which only covers a little part of a cell. Therefore, the UEs are commonly located out of the Rayleigh distance of BS, i.e. far-field region, where the channels can be approximately modeled based on the far-field planar-wave propagation model~\cite{Sherman'62'j}. Owing to the sparse scattering environment of mmWave frequency~\cite{ma'21'j}, the far-field channel can be modeled as
\begin{equation}
\label{eq: far-field channel}
\begin{aligned}
{\bf{h}}_k^{\rm{far}} = \sqrt{N} \alpha_{0} {\bf{a}}(\theta_0, \phi_0) + \sqrt{\frac{N}{L}} \sum_{l=1}^{L}\alpha_{l} {\bf{a}}(\theta_l, \phi_l),
\end{aligned}
\end{equation}
which contains one line-of-sight (LoS) path and $L$ non-line-of-sight (NLoS) paths. Parameters $\alpha_l$, $\theta_l$, and $\phi_l$ denote the complex path gain, elevation angle, and azimuth angle of the $l^{th}$ path, respectively. The index $l=0$ represents the LoS path, while $l \geq 1$ represents the NLoS paths. The complex gain for the LoS path $\alpha_0$ can be expressed as $\alpha_0 = \sqrt{\frac{\kappa}{\kappa+1}}$ while the channel gains for NLoS paths follow $\alpha_{l} \sim \mathcal{CN}(0, \sigma_{\alpha,l}^2)$ for $l \geq 1$, where $\sigma_{\alpha,l}^2 = \frac{1}{\kappa+1}$. The Rician factor $\kappa$ denotes the power ratio of LoS and NLoS paths. Due to the planar-wave propagation model, the far-field steering vector ${\bf{a}}(\theta_l, \phi_l)$ is only determined by the transmitting or receiving angles for a fixed array structure. For the widely adopted ULA, the steering vector corresponding to an $N$-element ULA along the $y$-axis can be written as 
\begin{equation}
\label{eq: ula}
\begin{aligned}
{\bf{a}}_{\rm{L}} (\phi) = \frac{1}{\sqrt{N}} \left[1, e^{jkd \sin\phi}, \cdots, e^{jk(N-1)d \sin\phi}\right]^T,
\end{aligned}
\end{equation}
where $k = \frac{2\pi}{\lambda}$ denotes the wavenumber, and $d$ denotes the spacing of adjacent antenna elements. It is worth noting that changing the elevation angle $\theta$ does not influence the element of ULA's steering vector~\cite{Ayach'14'j}. Thus, $\theta$ is omitted in the arguments of ${\bf{a}}_{\rm{L}}(\cdot)$ for simplicity. Similarly, for UPA with $N = N_1 \times N_2$ elements placed in the $yz$-plane, the steering vector can be written as~\cite{yan'22'j}
\begin{equation}
\label{eq: upa}
\begin{aligned}
{\bf{a}}_{\rm{P}} (\theta, \phi) & =  {\bf{a}}_{\rm{L}}^{(y)} \otimes {\bf{a}}_{\rm{L}}^{(z)}\\
& = \frac{1}{\sqrt{N}} \Big[1, \cdots, e^{jkd (n_1\sin\phi \sin\theta+n_2 \cos\theta)}, \\
& \cdots, e^{jkd((N_1-1) \sin\phi \sin\theta + (N_2-1)\cos\theta)} \Big]^T,
\end{aligned}
\end{equation}
where $N_1$ and $N_2$ denote the number of antenna elements on the $y$-axis and $z$-axis, respectively.
\par The beam steering vectors above are based on the planar-wave propagation assumption in the far-field region. However, since the number of antenna elements significantly increases in ELAA systems, the near-field region will also be enlarged. Specifically, as the number of antennas scales up, the near-field region determined by the Rayleigh distance $r_{\rm{RD}} = \frac{2D^2}{\lambda}$ for ELAA systems remarkably expands. For instance, the Rayleigh distance of a 3200-antenna array in~\cite{Balakrishnan'20'c} can reach about 200 meters. Most areas of a cell will be classified into the near-field region, where the electromagnetic field has to be modeled by spherical waves. Therefore, the classical beam steering vectors in~\eqref{eq: ula} and~\eqref{eq: upa} are no longer valid for near-field communications.

\subsection{Near-Field Channel Model}\label{sec: sys channel near}
\par To characterize the near-field channel for ELAA systems, the near-field channel model adopting the point scatterer assumption can be expressed as~\cite{Cui'22'tcom}
\begin{equation}
\label{eq: near-field channel}
\begin{aligned}
{\bf{h}}_k^{\rm{near}} = \sqrt{N} \alpha_{0} {\bf{b}}(r_0, \theta_0, \phi_0) + \sqrt{\frac{N}{L}} \sum_{l=1}^{L}\alpha_{l} {\bf{b}}(r_l, \theta_l, \phi_l),
\end{aligned}
\end{equation}
where ${\bf{b}}(r_l, \theta_l, \phi_l)$ denotes the near-field beam focusing vector, which focuses the signal energy at the location indexed by $(r_l, \theta_l, \phi_l)$. Again, $l=0$ represents the LoS channel component, and $l\geq 1$ represents the NLoS components. To emphasize the different properties of near-field beams, we term the single-path near-field channel as the \emph{beam focusing vector}, which is opposite to the classical \emph{beam steering vector} defined in~\eqref{eq: ula} and~\eqref{eq: upa} in the far-field region.

\par As shown in Fig.~\ref{img: channel ULA}, the near-field beam focusing vector for an $N$-element ULA can be expressed as
\begin{equation}
\label{eq: near field response}
\begin{aligned}
{\bf{b}}_{\rm{L}}(r, \phi) = \frac{1}{\sqrt{N}}\left[e^{-jk(r^{(-\widetilde{N})}-r)},\cdots,e^{-jk(r^{(\widetilde{N})}-r)}\right]^T,
\end{aligned}
\end{equation}
where $r^{(n)}$ denotes the distance between the scatterer (or UE) and the $n^{\rm th}$ antenna element, and $r$ denotes the distance between the scatterer (or UE) and the center of the array\footnote{Note that the polarization of each antenna also needs to be precisely configured to realize alignments of signals at receivers~\cite{myers2022near}. Since this paper mainly focuses on precoding designs, an ideal polarization configuration is assumed. Therefore, the factor of polarization is omitted throughout the paper.}. The maximum index is defined as $\widetilde{N} = \frac{N-1}{2}$, and $N$ is assumed to be odd. The distance term $r_l^{(n)}$ of the $l^{\rm th}$ path can be rewritten based on the spherical-wave propagation model as
\begin{equation}
\label{eq: near field distance term}
\begin{aligned}
r_{l}^{(n)} &= \sqrt{r_l^2-2ndr_l\sin\phi_l+n^2d^2} \\
& \mathop{\approx}\limits^{(a)} r_l\left(1+\frac{1}{2}\epsilon^2 - \frac{1}{8}\epsilon^2 \right)\\
& \mathop {\approx}\limits^{(b)} r_l - nd\sin\phi_l + \frac{n^2d^2}{2r_l}\cos^2\phi_l\\
& = r_l + \psi_{{\rm{L}}, r_l, \phi_l}^{(n)},
\end{aligned}
\end{equation}
where $\epsilon = -\frac{2nd\sin\phi_l}{r_l}+\frac{n^2d^2}{r_l^2}$ denotes the small quantity compared with $r_l$ and $\psi_{{\rm{L}}, r_l, \phi_l}^{(n)} = -nd\sin\phi_l + \frac{n^2d^2}{2r_l}\cos^2\phi_l$ denotes the difference of the propagation distance on the $l^{th}$ path compared with the ULA center. Approximation (a) is derived by the second-order Taylor series expansion $\sqrt{1+x} = 1 + \frac{x}{2} - \frac{x^2}{8} + \mathcal{O}(x^2)$. It has been proved in~\cite{Janaswamy'17'm} that, when communication distance exceeds Fresnel boundary $r_{\rm{FR}} = \frac{D}{2}\left(\frac{D}{\lambda}\right)^{1/3}$, the second-order expansion is commonly accurate enough in the near-field region. Moreover, approximation (b) is obtained by assuming that $r_l$ is much larger than $nd$ even in the near-field region, since communication distance is commonly larger than the array aperture. Therefore, we only keep the first and second-order of $\frac{nd}{r_l}$, which could obtain a concise but accurate approximation as~\cite{Cui'22'tcom}.

\begin{figure}[!t]
	\centering
	\setlength{\abovecaptionskip}{0.cm}
	\includegraphics[width=3in]{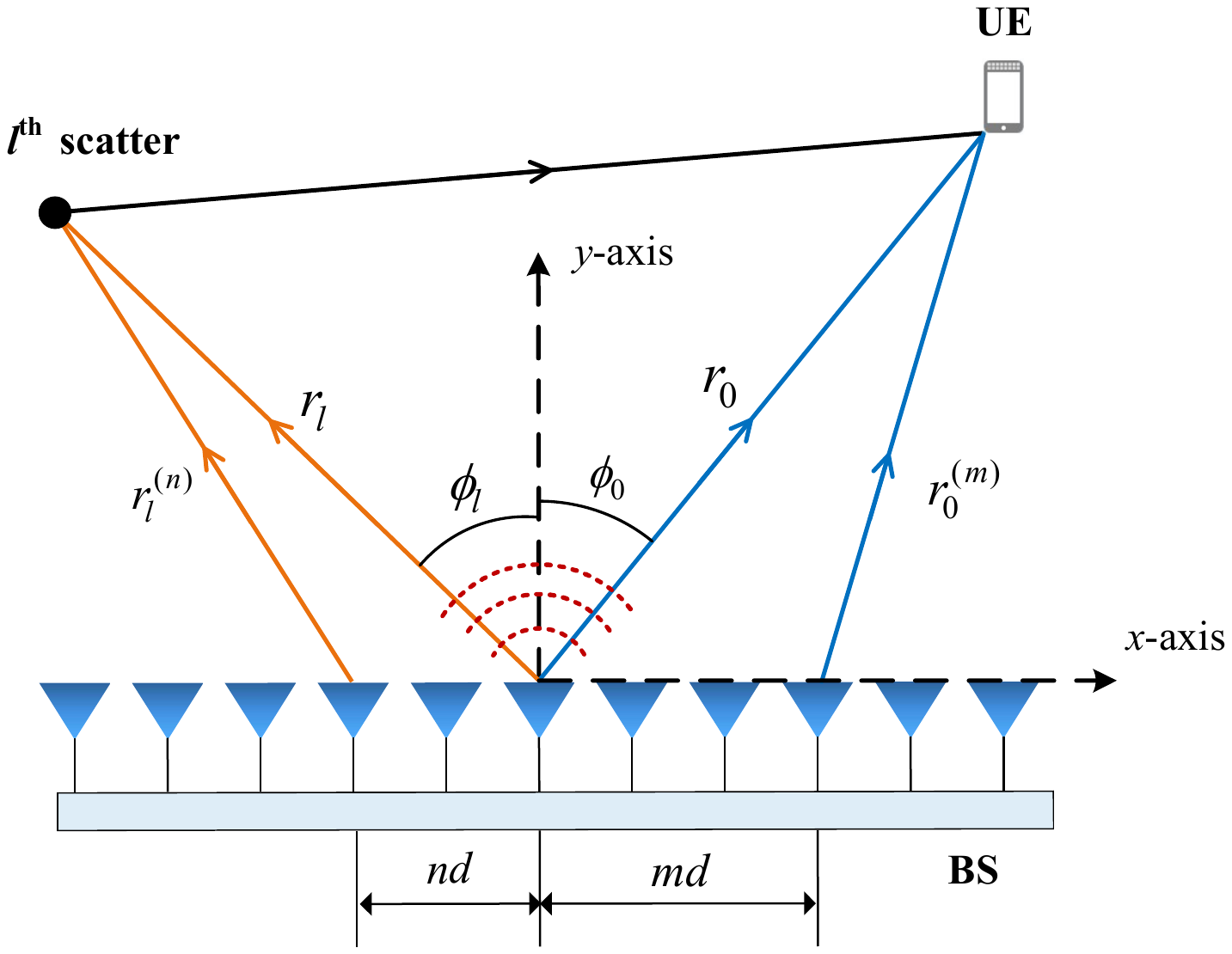}
	\caption{Near-field channel model for ULA communication systems.}
	\label{img: channel ULA}
    \vspace{-3mm}
\end{figure}

\begin{remark}
When the Taylor series expansion only keeps the first-order term as $\sqrt{1+x} \approx 1 + \frac{x}{2}$, the distance $r_{l}^{(n)}$ becomes proportional to index $n$, which makes near-field beam focusing vectors in~\eqref{eq: near field response} degenerate to far-field beam steering vectors in~\eqref{eq: far-field channel}. In other words, the linear phase of far-field beam steering vectors is a special case of non-linear phase of near-field beam focusing vector without higher-order information.
\end{remark}

\par For UPA with $N = N_1 \times N_2$ elements where antennas are deployed in two perpendicular directions, the beam focusing vectors are markedly different from those of ULA. Specifically, the near-field beam focusing vector for UPA can be obtained based on the spherical-wave propagation model as
\begin{equation}
\label{eq: near-field UPA channel}
\begin{aligned}
{\bf{b}}_{\rm{P}}(r, \theta, \phi) = \frac{1}{\sqrt{N}}\left[e^{-jk(r^{(-\widetilde{{\bf{\mu}}})}-r)},\cdots,e^{-jk(r^{(\widetilde{{\bf{\mu}}})}-r)}\right]^T,
\end{aligned}
\end{equation}
where $\widetilde{{\bf{\mu}}} = (\widetilde{N}_1, \widetilde{N}_2)$, $\widetilde{N}_1$ and $\widetilde{N}_2$ are similarly defined as $\widetilde{N}_i = \frac{N_i-1}{2}$. As shown in Fig.~\ref{img: near channel UPA}, the distance between the scatterer (or UE) and the $(n_1,n_2)$-element of the $l^{th}$ path is expressed as in (\ref{eq: near-field UPA distance term}) at the bottom, where $r_l$ denotes the distance between the center of UPA and the scatterer (or UE) of the $l^{\rm th}$ path and $\psi_{{\rm{P}}, r_l, \theta_l, \phi_l}^{(n_1, n_2)}$ denotes the difference of the propagation distance of the $l^{\rm th}$ path compared with the center of UPA. Again, the approximation (a) is derived by adopting a similar approximation as in~\eqref{eq: near field distance term}.
\begin{figure*}[hb]
\hrule
\begin{equation}
\label{eq: near-field UPA distance term}
\begin{aligned}
r_l^{(n_1,n_2)} &= \sqrt{(r_l\sin\theta_l\cos\phi_l)^2 + (r_l\cos\theta_l-n_1d)^2+(r_l\sin\theta_l \sin\phi_l-n_2d)^2} \\
& \mathop {\approx}\limits^{(a)} r_l - n_1d\cos\theta_l - n_2d\sin\theta_l\sin\phi_l + \frac{n_1^2d^2}{2r_l}(1-\cos^2\theta_l) + \frac{n_2^2d^2}{2r_l}(1-\sin^2\theta_l\sin^2\phi_l) - \frac{n_1n_2d^2\cos\theta_l\sin\theta_l\sin\phi_l}{r_l} \\
& = r_l + \psi_{{\rm{P}}, r_l, \theta_l, \phi_l}^{(n_1, n_2)}.
\end{aligned}
\end{equation}
\end{figure*}

\begin{remark}
It is worth noting that, the far-field steering vectors along the $y$- and $z$-axis can be decoupled. Thus, the far-field steering vector of UPA can be expressed as the Kronecker product of two steering vectors of ULA as in~\eqref{eq: upa}. While in the near-field region where the spherical-wave assumption is adopted, the phase term is coupled along the $y$- and $z$-axis in~\eqref{eq: near-field UPA distance term}, which brings fundamental distinction in the beam codebook design method, which is discussed in Section~\ref{sec: codebook}.
\end{remark}

\begin{figure}[!t]
	\centering
	\setlength{\abovecaptionskip}{0.cm}
	\includegraphics[width=3in]{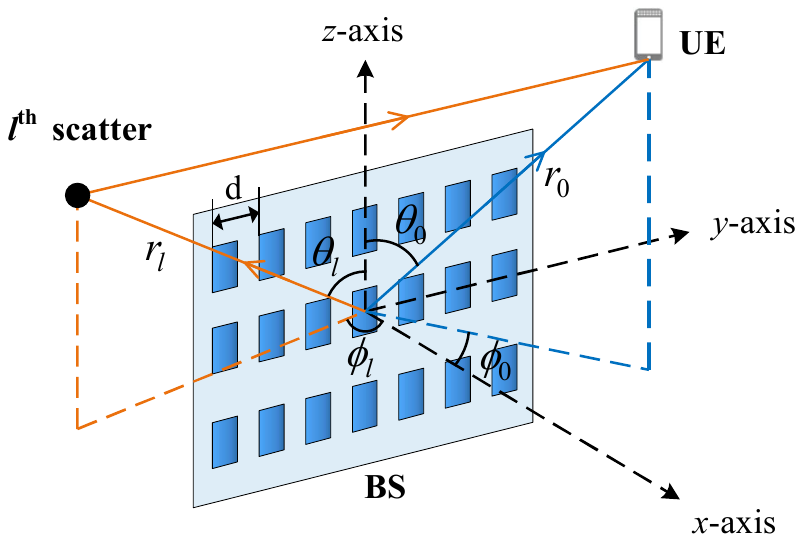}
	\caption{Near-field channel model for UPA communication systems.}
	\label{img: near channel UPA}
\end{figure}

To sum up, different from far-field channels based on planar-wave assumptions in~\eqref{eq: ula} and~\eqref{eq: upa}, the phase terms based on~\eqref{eq: near field distance term} and~\eqref{eq: near-field UPA distance term} are related to both spatial angle and distance. Therefore, the near-field channels for UEs located in the same angle but different distances are remarkably different. The UEs can be distinguished according to their spatial angle and distance, i.e., their location on the 2D (3D) space for the ULA (UPA) case, which is a fundamental change compared with far-field communications. In the following section, the spatial resolution of near-field beam focusing vectors in angular and distance domains will be discussed.

\section{Analysis of Asymptotic Orthogonality of Near-Field Beam Focusing Vectors}\label{sec: analysis}
To further analyze the property of near-field beam focusing vectors, the asymptotic orthogonality of near-field beam focusing vectors is investigated in this section. The orthogonality is revealed by investigating the correlation between different beam focusing vectors. We first review the asymptotic orthogonality of far-field beam steering vectors in the angular domain. Then, we show that the asymptotic orthogonality generalizes from the angular domain to both angular and distance domains in the near-field region for ULA systems. Finally, the proposition is generalized to UPA systems. The asymptotic orthogonality reveals the fundamental differences between far-field and near-field communications.

\subsection{Analysis of ULA Systems}\label{sec: analysis ULA}
We first consider the far-field communications for ULA systems. The correlation of steering vectors focusing on $\phi_l$ and $\phi_m$ can be formulated as
\begin{equation}
\label{eq: far gain}
\begin{aligned}
\left|{\bf{a}}_{\rm{L}}^H(\phi_l) {\bf{a}}_{\rm{L}}(\phi_m)\right| & = \frac{1}{N} \left|\sum_{n=-\widetilde{N}}^{\widetilde{N}} e^{jnkd(\sin\phi_m-\sin\phi_l)}\right| \\
& = \left| \Xi_N(kd(\sin\phi_m-\sin\phi_l)) \right|,
\end{aligned}
\end{equation}
where $\Xi_N(\alpha) = \sin\frac{N}{2}\alpha/(N\sin\frac{1}{2}\alpha)$ is the Dirichlet sinc function. According to~\eqref{eq: far gain}, the correlation of steering vectors achieves the maximum when $\phi_l = \phi_m$. If we consider two single-path UEs located at the same angle, the channels of the two UEs are identical and simultaneous transmissions can not be established through precoding. Otherwise, the correlation of steering vectors focusing on different angles tends to be orthogonal with infinite antennas as~\cite{Nehorai'95'j}
\begin{equation}
\label{eq: far infinity}
\begin{aligned}
\lim_{N \to +\infty} |{\bf{a}}^H(\phi_l) {\bf{a}}(\phi_m)| = 0 \ , \  \phi_l \neq \phi_m.
\end{aligned}
\end{equation}

Therefore, the spatial angular resolution of BS tends to infinity as the number of antennas increases. Owing to the angular orthogonality, BS is able to distinguish different channel components and multiplex different data streams to different UEs. Therefore, the angular orthogonality of the far-field beam steering vectors contributes to the SDMA or beam division multiple access (BDMA) schemes \cite{Xiao'15'j}. In classical massive MIMO communications, SDMA or BDMA schemes have been thoroughly investigated. Nevertheless, the far-field beam steering vector is no longer valid in the near-field region, where the propagation model changes from planar-wave models into spherical-wave models.

\par Based on the near-field channel focusing vector in~\eqref{eq: near field response}, the correlation of two beam focusing vectors corresponding to the location of $(r_l, \phi_l)$ and $(r_p, \phi_p)$ is written as
\begin{equation}
\label{eq: near gain}
\begin{aligned}
\left|{\bf{b}}_{\rm{L}}^H(r_l, \phi_l) {\bf{b}}_{\rm{L}}(r_m, \phi_m)\right| = \frac{1}{N} \left| \sum_{n=-\widetilde{N}}^{\widetilde{N}} e^{jk (\psi_{{\rm{L}}, r_l, \phi_l}^{(n)} - \psi_{{\rm{L}}, r_m, \phi_m}^{(n)})} \right|,  
\end{aligned}
\end{equation}
where $\psi_{{\rm{L}}, r_l, \phi_l}^{(n)}$ and $\psi_{{\rm{L}}, r_m, \phi_m}^{(n)}$ are defined as in~\eqref{eq: near field distance term}.
\par Unlike the far-field beam steering vectors which are only dependent on spatial angles, the near-field beam focusing vector is related to both spatial angle and distance. According to {\bf{Lemma} 1} in~\cite{Cui'22'tcom}, the correlation of two near-field beam focusing vectors corresponding to the same angle but different distances can be characterized with the following lemma. 
\begin{lemma}
\label{lemma1}
The correlation of near-field beam focusing vectors can be approximated as follows
\begin{equation}
\label{eq: near field correlation}
\begin{aligned}
\left|{\bf{b}}^H(r_l, \phi) {\bf{b}}(r_m, \phi) \right| \approx \left| \frac{C(\beta)+jS(\beta)}{\beta} \right|, 
\end{aligned}
\end{equation}
where $\beta = \sqrt{\frac{d^2\cos^2\phi}{2\lambda}|\frac{1}{r_l}-\frac{1}{r_m}|}N$, $C(\cdot)$ and $S(\cdot)$ denote the Fresnel functions written as $C(x) = \int_0^x \cos(\frac{\pi}{2}t^2){\rm{d}}t$ and $S(x) = \int_0^x \sin(\frac{\pi}{2}t^2){\rm{d}}t$ \cite{Sherman'62'j}.
\end{lemma}
On a basis of {\bf{Lemma}~\ref{lemma1}}, the asymptotic orthogonality of ULA systems in the distance domain could be directly derived as follows.
\begin{corollary}[ULA Asymptotic Orthogonality in Distance Domain]
\label{coro1}
Near-field beam focusing vectors corresponding to the same angle and different distances are asymptotically orthogonal with the increasing number of antennas for ULA systems, which is to say
\begin{equation}
\label{eq: near lim}
\begin{aligned}
\lim_{N \to +\infty} \left|{\bf{b}}^H(r_l, \phi) {\bf{b}}(r_m, \phi) \right|  = 0,\ {\rm{for}}\ r_l \neq r_m.
\end{aligned}
\end{equation}
\end{corollary}
\begin{proof}
According to~\eqref{eq: near field correlation}, when the number of antennas $N$ tends to infinity, the numerator $C(\beta) + jS(\beta)$ converges to $0.5 + 0.5j$ and the denominator $\beta$ tends to $+\infty$~\cite{boersma1960computation}. Therefore, the correlation converges to 0, which proves the asymptotic orthogonality in the distance domain.
\end{proof}
To verify the asymptotic orthogonality in distance domain with a numerical example, the correlation of beam focusing vectors corresponding to $(5\,{\rm{m}}, \pi/6)$ and $(15\,{\rm{m}}, \pi/6)$ is plotted in Fig.~\ref{img:Cor_ula}. The frequency is set to $30$ GHz and the antennas are half-wavelength spaced. It is shown that the correlation significantly decreases as the number of antennas increases. Moreover,~\eqref{eq: near field correlation} can well approximate the accurate correlation\footnote{It is worth noting that the blue solid line in Fig.~\ref{img:Cor_ula} is obtained with the approximation of $r_l$ in~\eqref{eq: near field distance term}. One may notice that when the number of antennas tends to infinity, the near-field assumption does not hold anymore and the approximation in~\eqref{eq: near field distance term} is no longer valid since higher-order terms of Taylor series occur. However, we mainly focus on the asymptotic trends assuming that near-field assumptions are valid, which is similar to the asymptotic orthogonality in the angular domain in the far-field region \cite{Xiao'15'j}.}.

\begin{figure}[!t]
	\centering
	\setlength{\abovecaptionskip}{0.cm}
	\includegraphics[width=3in]{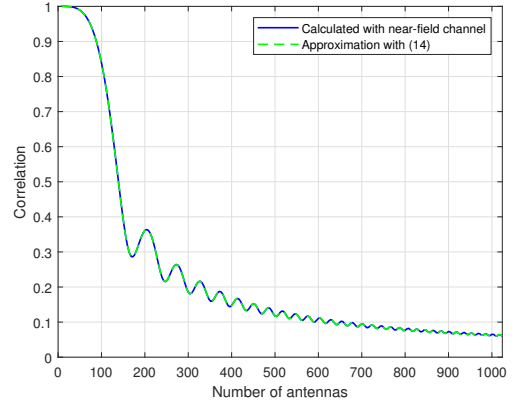}
	\caption{Correlation with increasing antennas for ULA systems.}
	\label{img:Cor_ula}
\end{figure}

Furthermore, we prove a more general 2D orthogonality in both angular and distance domains in the near-field region for ULA systems in the following corollary.
\begin{corollary}[ULA Asymptotic Orthogonality in 2D Domain]
\label{coro2}
Near-field beam focusing vectors corresponding to any different angles or different distances are also asymptotically orthogonal with the increasing number of antennas for ULA systems, which is to say
\begin{equation}
\label{eq: near lim 2D}
\begin{aligned}
\lim_{N \to +\infty} \left|{\bf{b}}_{\rm{L}}^H(r_l, \phi_l) {\bf{b}}_{\rm{L}}(r_m, \phi_m)\right| = 0,\,{\rm{for}}\,r_l \neq r_m\ {\rm{or}}\ \phi_l \neq \phi_m.
\end{aligned}
\end{equation}
\end{corollary}

\begin{proof}
The proof is provided in {\bf Appendix~\ref{app: coro2}}.
\end{proof}
\par The corollary reveals that the angular orthogonality in the far-field region generalizes into the 2D orthogonality in the near-field region. The orthogonality indicates a stronger potential to establish simultaneous transmissions through precoding in ELAA communications. Therefore, it brings possibilities for a novel multiple access scheme, which will be discussed in Section~\ref{sec: LDMA}.
\par So far, only the ULA systems are investigated. The asymptotic orthogonality of beam focusing vectors for UPA systems is discussed in the following subsection.

\subsection{Analysis of UPA Systems}\label{sec: analysis UPA}
Different from near-field beam focusing vectors of ULA, the last term in~\eqref{eq: near-field UPA distance term} indicates a coupling of the indices in two directions. Based on the near-field beam focusing vector of UPA systems defined in~\eqref{eq: near-field UPA channel}, the asymptotic orthogonality could be similarly obtained through the following lemma.
\begin{lemma}[UPA Distance Asymptotic Orthogonality]
\label{lemma2}
Near-field beam focusing vectors focusing on the same elevation and azimuth angle but different distances are asymptotically orthogonal with the increasing number of antennas in UPA systems. Without loss of generality, when $N_1 \to \infty$ we have
\begin{equation}
\label{eq: near lim 2D UPA}
\begin{aligned}
\lim_{N_1 \to +\infty} \left|{\bf{b}}_{\rm{P}}^H(r_l, \phi, \theta) {\bf{b}}_{\rm{P}}(r_m, \phi, \theta)\right| = 0,\ {\rm{for}}\ r_l \neq r_m.
\end{aligned}
\end{equation}
\end{lemma}
\begin{proof}
A proof sketch is provided in {\bf Appendix~\ref{app: lemma2}}.
\end{proof}

Lemma 2 reveals a stronger conclusion compared with the ULA cases discussed in {\bf{Corollary~\ref{coro1}}}. It indicates that in UPA communications, as long as the number of antennas at one side ($N_1$ or $N_2$) is very large, the orthogonality of two different near-field beam focusing vectors can be guaranteed.
\par So far, we have revealed the asymptotic orthogonality for both ULA and UPA systems. The asymptotic orthogonality of beam steering (focusing) vectors generalizes from the angular domain in the far-field region to both angular and distance domains in the near-field region. The improved spatial orthogonality in the 2D or 3D space, i.e., the interference of the incoming electromagnetic waves radiated from different angles or different distances is limited, could contribute to interference suppression in multi-user ELAA communications. Therefore, a novel multiple access scheme based on the UE location is promising to increase the system performance in ELAA communication systems.

\section{Near-Field Location Division Multiple Access Scheme}\label{sec: LDMA}
In this section, the LDMA scheme for ELAA communication systems is introduced. Then the near-field beam codebook for the initial access procedure is proposed, which is critical to maximize the received signal power and manage the inter-user interference in the LDMA scheme. 

\subsection{LDMA Framework}\label{sec: ldma framework}
As illustrated in Section~\ref{sec: analysis}, the orthogonality of near-field channels in both angular and distance domains reveals an increased precision of beamforming. The near-field beamforming is capable of focusing the energy on a specific location rather than a specific angle. The fundamental change of beamforming property indicates that, the far-field steering beams can be replaced by near-field focusing beams to improve the received signal and suppress the interference from other UEs. 
\par To achieve this goal, the concept of LDMA is proposed in practical hybrid precoding communications. The core idea is to employ near-field location-dependent beam focusing vectors as analog precoders to serve different UEs located in different locations. Compared with classical far-field SDMA schemes~\cite{Heath'15'j}, the proposed LDMA scheme is able to utilize the focusing of near-field beams rather than the steering property of far-field beams~\cite{cui'22'm}, providing a different method to efficiently take advantage of spatial resources.

\subsection{Codebook Design}\label{sec: codebook}
In the proposed LDMA scheme, precoding is essential to suppress the inter-user interference to enable parallel transmissions for different UEs. Since the digital precoder is relatively easy to design after acquiring the effective channel, a critical problem is the design of analog precoder. To reduce the overhead of initial access, a finite-size codebook is often preferred to design the analog precoding.

The finite-size codebook design methods can be mainly categorized into two classes, i.e. quantized beamforming codebooks and beamsteering codebooks (although precisely the focusing property is employed in near-field, we still call it beamsteering codebook for consistency). Generally speaking, quantized beamforming codebooks are designed to minimize the quantized maximum correlation of codes under specific distribution assumptions, such as Grassmannian codebooks~\cite{Strohmer'03'j}, random vector quantization (RVQ) codebooks~\cite{Love'07'j}, and generalized Lloyd codebooks~\cite{Giannakis'06'j}. Nevertheless, the complexity of designing quantized beamforming codebooks significantly increases with the increasing number of antennas. Also, the distribution for designing such codebooks is often hard to obtain, since the analysis of distribution of near-field beam focusing vectors over the high-dimension hypersphere is extremely complicated. Another category, i.e. beamsteering codebook, is preferred to address the problem~\cite{Heath'15'j}, where the beam steering vectors are employed to construct the codebook. Since the beam steering vectors can be viewed as the matched filters of the single-path channel, the beamsteering codebook is quite suitable for designing the analog precoder which aims to enlarge the received signal during the transmission.

In massive MIMO systems, the DFT codebooks are often adopted to perform analog precoding. The DFT codebook is constructed with a series of selected steering vectors, corresponding to angles that are uniformly sampled on the sine function value over a specific angular range. However, since the electromagnetic field changes from planar waves into spherical waves, the far-field codebooks are no longer valid in the near-field region. The mismatch of far-field codebook and near-field channel seriously degrades beamforming gain~\cite{Heath'22'j}.

To overcome the degradation, some near-field codebook design methods have been proposed. In \cite{han'20'j} a uniform sampling grid on the $x$- and $y$-axis is adopted for designing the dictionary channel matrix to perform compressive sensing. Similarly, a uniform sampling method was performed to construct the dictionary matrix for RIS-aided communications~\cite{Xiu'22'j}. On this basis, the correlation of near-field beams is proved to change non-uniformly along the distance dimension in~\cite{Cui'22'tcom}, where an efficient non-uniform sampling method in the polar domain is proposed for near-field channel estimation. However, the proposed codebook is based on ULA structures and can not be directly applied to UPA systems. To enable near-field communications with UPA, the near-field beamsteering codebook sampled in the 3D space for UPA systems is desired to be investigated.

\par First, if we consider a planar-wave propagation model, the correlation of far-field beam steering vectors corresponding to angles $(\theta_p, \phi_p)$ and $(\theta_q, \phi_q)$ can be expressed as
\begin{equation}
\label{eq: angular sampling}
\begin{aligned}
f_{l, m} = &\Bigg| \frac{1}{N} \sum_{n_1} \sum_{n_2} \exp \bigg( jkn_1d(\cos\theta_m-\cos\theta_l) \\
& + jkn_2d(\sin\theta_m\sin\phi_m - \sin\theta_l\sin\phi_l) \bigg) \Bigg| \\
 = & \left| \Xi_{N_1}(kd(\cos\theta_m-\cos\theta_l)) \right| \\
& \times \left|\Xi_{N_2}(kd(\sin\theta_m\sin\phi_m-\sin\theta_l\sin\phi_l)) \right|.
\end{aligned}
\end{equation}
Therefore, the angular samples can be chosen as
\begin{equation}
\label{eq: angular sampling on angle}
\begin{aligned}
\cos\theta_{n_1} = & \frac{2n_1-N_1+1}{N_1},\ n_1=0,1,\cdots,N_1-1 \\
\sin\theta_{n_2}\sin\phi_{n_2} = & \frac{2n_2-N_2+1}{N_2},\ n_2=0,1,\cdots,N_2-1,
\end{aligned}
\end{equation}
which ensures the orthogonality of beam steering vectors. Compared with the far-field DFT-based codebook which is only determined by the elevation and azimuth angle, an additional problem arises in the near-field region, which is how to perform sampling in the radius dimension. As illustrated in Section~\ref{sec: analysis UPA}, the difficulties to design codebooks for UPA systems lie in the three quadratic terms as in~\eqref{eq: near-field UPA distance term}. To resolve this problem, we first show that when $N = N_1 \times N_2$ is not extremely large, the bilinear quadratic term $\frac{n_1n_2d^2\cos\theta_l\sin\theta_l\sin\phi_l}{r_l}$ in~\eqref{eq: near-field UPA distance term} can be neglected. Thus, the distance term in~\eqref{eq: near-field UPA distance term} can be approximated as
\begin{equation}
\label{eq: cross term approximation}
\begin{aligned}
{\widetilde{r}}_{l}^{(n_1,n_2)} &\approx r_l - n_1d\cos\theta_l - n_2d\sin\theta_l\sin\phi_l \\
& + \frac{n_1^2d^2}{2r_l}(1-\cos^2\theta_l) + \frac{n_2^2d^2}{2r_l}(1-\sin^2\theta_l\sin^2\phi_l).
\end{aligned}
\end{equation}

With the approximated distance term, the approximated beamforming vector ${\widetilde{\bf{b}}} (r_l, \theta_l, \phi_l)$ can be constructed similarly to~\eqref{eq: near field response}. Next, we show that the performance loss of beamforming using approximated beamforming vector is negligible through the following {\bf Lemma~\ref{lemma3}}.
\begin{lemma}
\label{lemma3}
The beamforming gain using approximated beamforming vector without the bilinear quadratic term can be written as
\begin{equation}
\label{eq: approximate cross term}
\begin{aligned}
{\widetilde{g}} &= \left| {\widetilde{\bf{b}}}^H(r, \theta, \phi) {\bf{b}}(r, \theta, \phi) \right| \\
&= \frac{1}{N} \left | \sum_{n_1} \sum_{n_2} \exp\left(\frac{jkn_1n_2d^2\cos\theta\sin\theta\sin\phi}{r}\right) \right| \\
&\approx \left| \frac{1}{\eta}{\rm{Si}}(\eta) \right|,
\end{aligned}
\end{equation}
where ${\rm{Si}}(x)$ denotes the sine integral ${\rm{Si}}(x) = \int_0^x \frac{\sin t}{t}\mathrm{d}t$ and $\eta = \frac{N_1N_2kd^2\cos\theta\sin\theta\sin\phi}{4r}$.
\end{lemma}
\begin{proof}
With accurate distance term defined in~\eqref{eq: near-field UPA distance term} and approximated distance term defined in~\eqref{eq: cross term approximation}, the only difference lies in the bilinear quadratic term. By assuming $\varsigma = \frac{kd^2\cos\theta\sin\theta\sin\phi}{r}$, the beamforming gain can be formulated as
\begin{equation}
\label{eq: beamforming gain upa}
\begin{aligned}
{\widetilde{g}} =& \left |{\widetilde{\bf{b}}}^H(r, \theta, \phi) {\bf{b}}(r, \theta, \phi) \right| \\
=& \frac{1}{N} \left | \sum_{n_1=-\widetilde{N}_1}^{\widetilde{N}_1} \sum_{n_2=-\widetilde{N}_2}^{\widetilde{N}_2} e^{j\varsigma n_1n_2} \right| \\
 \mathop{\approx}\limits^{(a)}& \frac{1}{N} \left| \int_{-N_1/2}^{N_1/2} \int_{-N_2/2}^{N_2/2} e^{j\varsigma n_1n_2} {\rm{d}}n_1 {\rm{d}}n_2 \right|\\
 =& \frac{1}{N} \left| \iint \cos(\varsigma n_1n_2) {\rm{d}}n_1 {\rm{d}}n_2 + j\iint \sin(\varsigma n_1n_2) {\rm{d}}n_1 {\rm{d}}n_2 \right|\\
= & \frac{1}{N} \left| \left[\int \frac{\sin(\varsigma n_1n_2)}{\varsigma n_1} {\rm{d}}n_1 - j\int \frac{\cos(\varsigma n_1n_2)}{\varsigma n_1} {\rm{d}}n_1 \right]_{n_2 = -\frac{N_2}{2}}^{\frac{N_2}{2}} \right|,\\
\end{aligned}
\end{equation}
Approximation (a) is obtained by replacing summation with integral and the limit of integral is omitted for simplicity. Owing to the odd function property of $\frac{\cos(x)}{x}$, the latter integral equals $0$. Then, the approximated beamforming gain is reformulated as 
\begin{equation}
\label{eq: beamforming gain upa2}
\begin{aligned}
{\widetilde{g}} = \frac{1}{N} &\left| \int_{-N_1/2}^{N_1/2} \frac{2}{\varsigma n_1} \sin(\varsigma {\widetilde{N}}_2n_1) {\rm{d}}n_1 \right|\\
 = & \left| \frac{1}{\eta} \int_{0}^{\eta} \frac{\sin(t)}{t} {\rm{d}}t \right| = \left| \frac{1}{\eta} {\rm{Si}}(\eta) \right|,\\
\end{aligned}
\end{equation}
which completes the proof.
\end{proof}

\begin{remark}
{\bf{Lemma}~\ref{lemma3}} has characterized the beamforming loss introduced by neglecting the intractable bilinear quadratic term. We shall show that the beamforming loss is commonly negligible with a numerical example. If we aim to ensure a beamforming loss less than threshold $\Delta$, using the monotonic descent property of $\frac{{\rm{Si}}(\eta)}{\eta}$ in the range $[0,\pi]$, the variable $\eta$ has to satisfy $\eta < \eta(\Delta)$ where $\frac{{\rm{Si}}(\eta(\Delta))}{\eta(\Delta)} = \Delta$. Then, according to the definition of $\eta$, a sufficient condition for $r$ is $\frac{N_1N_2kd^2\cos\theta\sin\theta\sin\phi}{4\eta(\Delta)} \leq \frac{\sqrt{3} \pi}{8\lambda \eta(\Delta)}N_1N_2d^2 \leq r$, which is commonly less than Rayleigh distance $r_{\rm{RD}} = \frac{2(D_1^2+D_2^2)}{\lambda}$. For instance, if we consider a $256 \times 16$ UPA operating at 30 GHz and threshold $\Delta = 5\%$, the minimum required distance is $7.24$ m. This indicates that when communication distance $r$ exceeds $7.24$ m, the beamforming loss by neglecting the bilinear term is at most $5\%$. Therefore, the omission of the bilinear term only introduces limited beamforming loss, but it brings convenience to the UPA codebook design, which is discussed as follows.
\end{remark} 

\begin{lemma}
\label{lemma4}
\par If the approximated near-field beam focusing vector is adopted, the correlation of two beam focusing vectors focusing in the same direction but different distances can be formulated as
\begin{equation}
\label{eq: correlation in distance upa}
\begin{aligned}
& ~~~\left|{\widetilde{\bf{b}}}^H(r_l, \theta, \phi){\widetilde{\bf{b}}}(r_m, \theta, \phi)\right|\\
&= \left|\sum_{n_1} \sum_{n_2} e^{jk(n_1^2d^2(1-\cos^2\theta) + n_2^2d^2(1-\sin^2\theta\sin^2\phi))(\frac{1}{2r_l}-\frac{1}{2r_m})} \right|\\
&\approx \left| G(\beta_1) G(\beta_2) \right|,
\end{aligned}
\end{equation}
where $G(\beta) = \frac{C(\beta)+jS(\beta)}{\beta}$, $C(\cdot)$ and $S(\cdot)$ denote the Fresnel functions written as $C(x) = \int_0^x \cos(\frac{\pi}{2}t^2){\rm{d}}t$ and $S(x) = \int_0^x \sin(\frac{\pi}{2}t^2){\rm{d}}t$. In addition, $\beta_1 = \sqrt{\frac{N_1^2d^2(1-\cos^2\theta)}{2\lambda}|\frac{1}{r_l}-\frac{1}{r_m}|}$ and $\beta_2 = \sqrt{\frac{N_2^2d^2(1-\sin^2\theta\sin^2\phi)}{2\lambda}|\frac{1}{r_l}-\frac{1}{r_m}|}$.
\end{lemma}
\begin{proof}
The proof is provided in Appendix~\ref{app: proof of lemma4}.
\end{proof}
This lemma provides a distance sampling method, where the correlation of two vectors concentrating on $(r_l, \theta, \phi)$ and $(r_m, \theta, \phi)$ can be analytically approximated. Noting that $\beta_1$ and $\beta_2$ share some common components for fixed spatial angles $(\theta, \phi)$, we define a new function as
\begin{equation}
\label{eq: correlation in distance upa one-variable}
\begin{aligned}
\bar{G}(\beta_0) = G(\beta_1) G(\beta_2),
\end{aligned}
\end{equation}
where $\beta_0 = \sqrt{\frac{d^2}{2\lambda}|\frac{1}{r_l}-\frac{1}{r_m}|}$. Thus, $\beta_1$ and $\beta_2$ can be determined by $\beta_0$ as $\beta_1 = N_1\beta_0\sqrt{1-\cos^2\theta}$ and $\beta_2 = N_2\beta_0\sqrt{1-\sin^2\theta\sin^2\phi}$, respectively. To reveal the trend of $|\bar{G}(\beta_0)|$, a numerical integration of $|\bar{G}(\beta_0)|$ against $\beta_0$ is performed and the result is shown in Fig.~\ref{img: distance sampling UPA}.

\begin{figure}[!t]
	\centering
	\setlength{\abovecaptionskip}{0.cm}
	\includegraphics[width=3in]{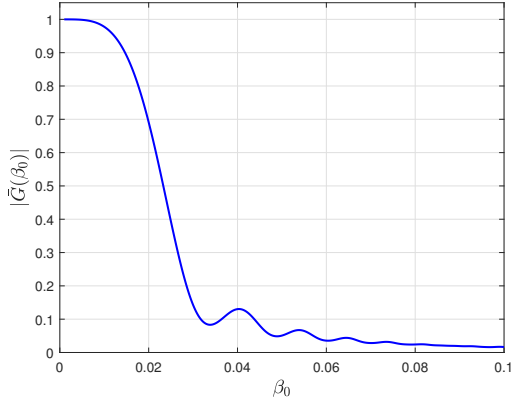}
	\caption{Numerical results of $|\bar{G}(\beta_0)|$ against $\beta_0$ for a $64 \times 64$ UPA system working at $30$ GHz sampled at $\phi=\pi/6$ and $\theta=\pi/3$.}
	\label{img: distance sampling UPA}
\end{figure}

Due to the descending trend of $|\bar{G}(\beta_0)|$, $|\bar{G}(\beta_0)| \leq \Delta$ can be approximately ensured if $\beta_0 \geq \beta_{\Delta}$, where $\bar{G}(\beta_{\Delta}) = \Delta$. Therefore, the correlation of beam focusing vectors focusing on the same angle can be lower than a predefined threshold $\Delta$ by controlling the sampling density $\left|\frac{1}{r_l}-\frac{1}{r_m}\right|$ as
\begin{equation}
\label{eq: distance sampling density}
\begin{aligned}
&\left|\frac{1}{r_l}-\frac{1}{r_m}\right| \geq \frac{2\lambda \beta_{\Delta}^2}{d^2}.
\end{aligned}
\end{equation}
It is worth noting that although $G(\cdot)$ can not be straightforwardly calculated analytically, it can be numerically calculated and stored for repeated usage. Therefore, even though $\beta_{\Delta}$ has to be searched for different $\theta$ and $\phi$, we only need to calculate $G(\cdot)$ once, which avoids repeatedly performing numerical integral. Finally, an efficient distance sampling method can be obtained with 
\begin{equation}
\label{eq: distance sampling on distance}
\begin{aligned}
r_s(\theta, \phi) = \frac{1}{s} \frac{d^2}{2\lambda \beta_{\Delta}^2},\ s=0,1,2,\cdots,
\end{aligned}
\end{equation}
where $s$ denotes the sampling index. The index $s=0$ represents the beamforming vector focusing the energy on infinite distance, where the near-field beamforming vector degenerates into the far-field beamforming vector. Therefore, the far-field UPA codebook can be viewed as a special case of the proposed near-field spherical-domain codebook and is contained in the proposed near-field codebook. When mixed-field communications happen, i.e. users having near-field or far-field channels are mixed, the proposed codebook could still work by allocating near (far) -field beams to near (far) -field users.
\par The proposed spherical-domain codebook design method is summarized in {\bf{Algorithm}~\ref{alg:1}}. To characterize the locations of the samples, the location with $s=1$ for each direction is shown in Fig.~\ref{img: distance sampling 3D}, where a $64 \times 64$ UPA operating at $30$ GHz is employed. In addition, each beamforming vector of the designed codebook $\mathcal{W}$ is of constant modulus, which coincides with the requirement of analog phase shifters in the hybrid precoding architectures. Therefore, the near-field UPA codebook is naturally applicable to analog precoding in multi-user MIMO systems.
\begin{figure}[!t]
	\centering
	\subfigbottomskip=2pt
	\subfigcapskip=-5pt
	\subfigure[Sampling points for $s=1$]{
		\label{3d.sub.1}
		\includegraphics[width=2.5in]{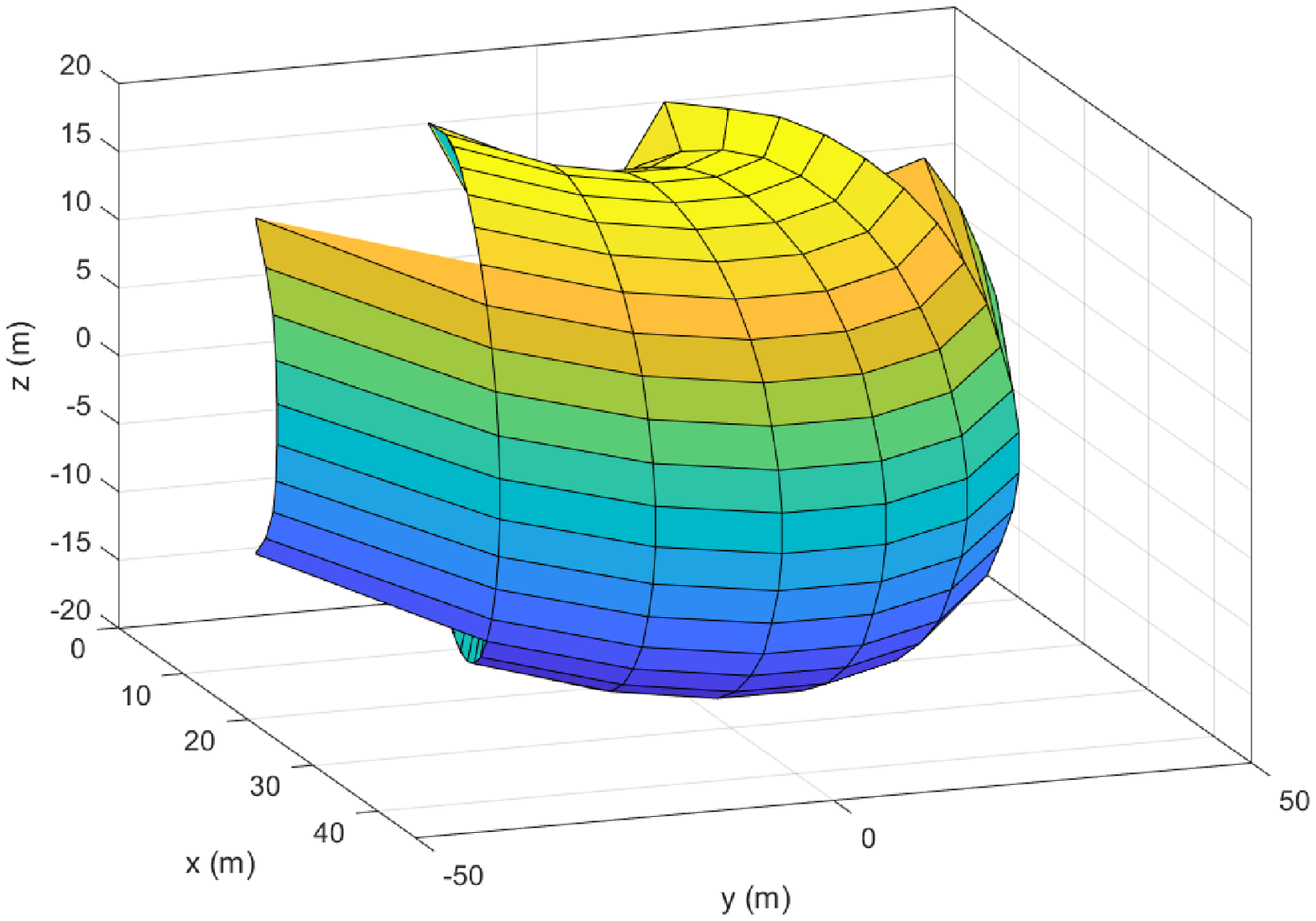}}
	\subfigure[Sampling points for $\theta=0$]{
		\label{3d.sub.2}
		\includegraphics[width=1.5in]{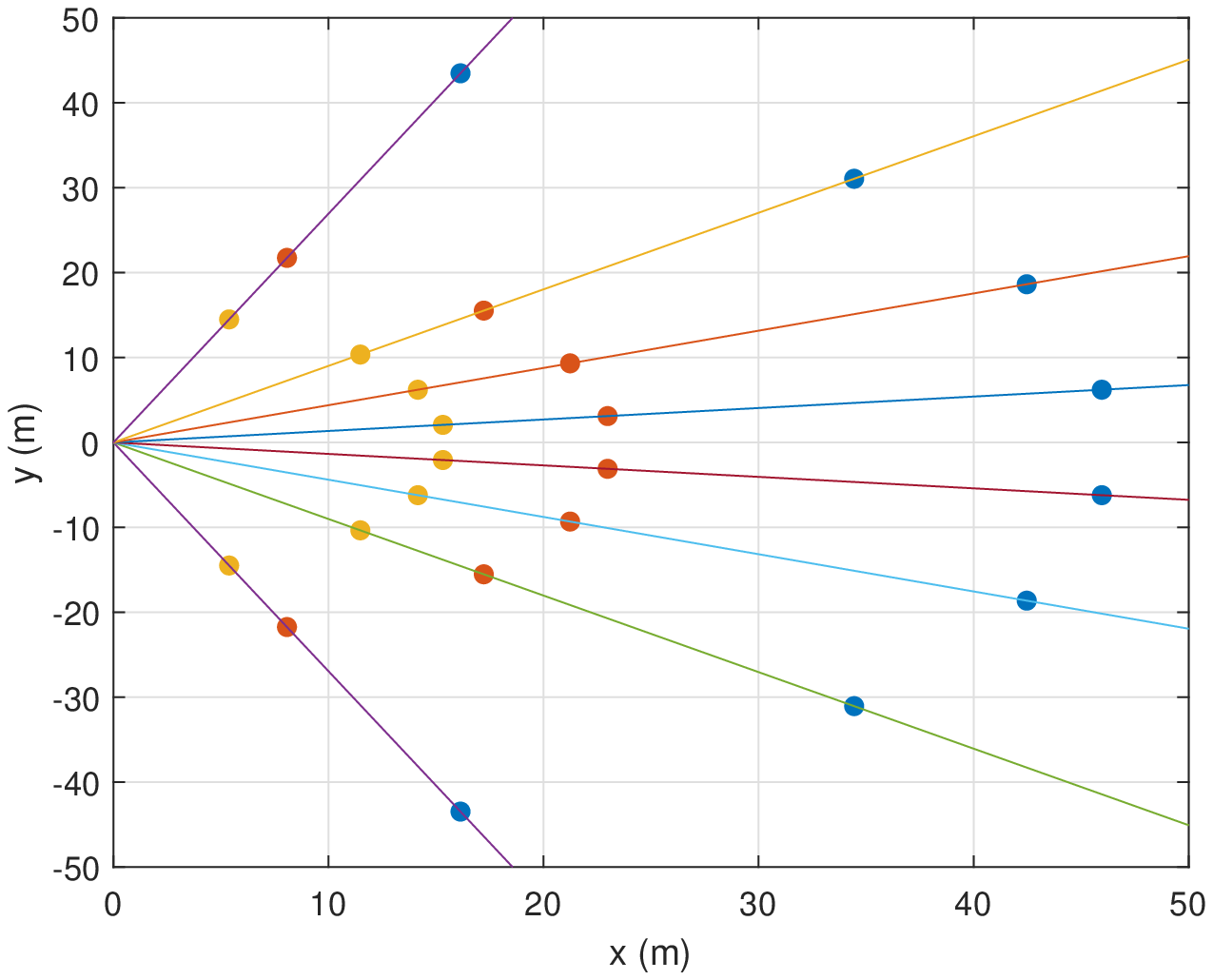}}
	\subfigure[Sampling points for $\phi=0$]{
		\label{3d.sub.3}
		\includegraphics[width=1.5in]{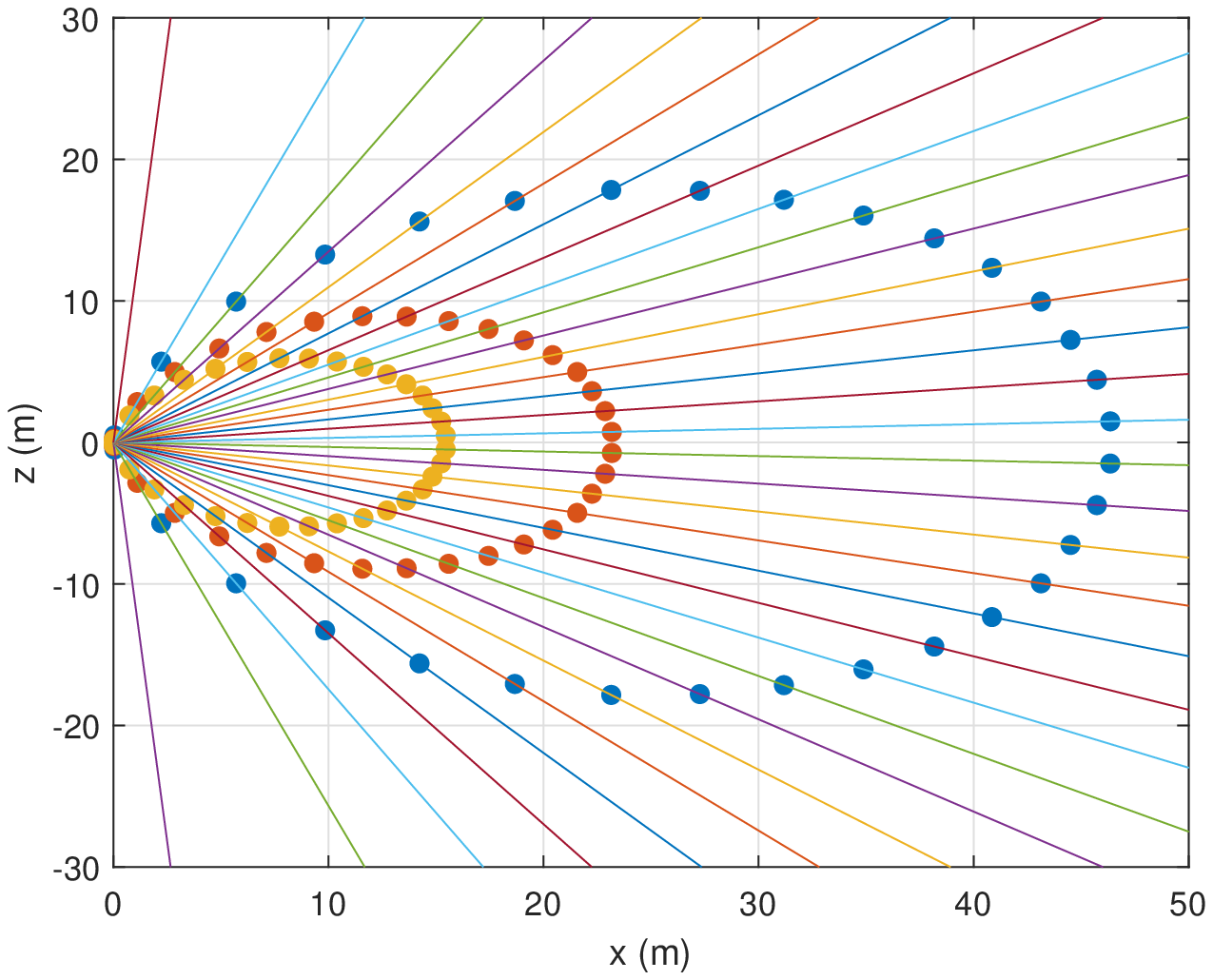}}
	\caption{An example of the sampling points according to the spherical-domain codebook. The intersections of black lines in (a) denote the sampling points corresponding to focusing position of near-field beams. The colored lines and colored dots in (b), (c) denote the sampling angles and sampling points, respectively.}
	\label{img: distance sampling 3D}
\end{figure}

\begin{algorithm}[!t] 
	\caption{Spherical-domain codebook design procedure.} 
	\label{alg:1} 
	\begin{algorithmic}[1] 
		\REQUIRE ~ 
		Minimum allowable distance $\rho_{\rm min}$, threshold $\Delta$, antennas $N_1, N_2$, antenna spacing $d$, wavelength $\lambda$
		\ENSURE ~ 
        Near-field Codebook ${\mathcal{W}}$
		\STATE $s=0$;
		\WHILE{$s=0$ or $\mathop{\rm{max}}\limits_{n1,n2}~r_{s}(\theta_{n1},\phi_{n2}) \geq \rho_{\rm min}$}
		\FOR{$n_1 = 0,1,\cdots,N_1-1$}
		\STATE Select $\theta_{n1}$ according to~\eqref{eq: angular sampling on angle};
		\FOR{$n_2 = 0,1,\cdots,N_2-1$}
		\STATE Select $\phi_{n2}$ according to~\eqref{eq: angular sampling on angle};
		\STATE Search for $\beta_{\Delta}$ satisfying $\bar{G}(\beta_{\Delta}) = \Delta$;
		\STATE Determine $r_{s}(\theta_{n1},\phi_{n2}) = \frac{1}{s} \frac{d^2}{2\lambda\beta_{\Delta}^2}$ by~\eqref{eq: distance sampling on distance};
		\ENDFOR
		\ENDFOR
        \STATE $\mathcal{W}_s = \left\{{\bf{b}}(r_{s}(\theta_{n_1},\phi_{n_2}), \theta_{n_1}, \phi_{n_2})\right\}_{n_i = 0,1,\cdots,N_i-1}$;
        \STATE $S = s$, $s=s+1$;
        \ENDWHILE
		\STATE $\mathcal{W} = [\mathcal{W}_0, \mathcal{W}_1, \cdots, \mathcal{W}_S ]$;
		\RETURN $\mathcal{W}$.
	\end{algorithmic}
\end{algorithm}

\subsection{LDMA Communication Procedure}\label{sec: ldma procedure}
\par The detailed communication procedure of LDMA in TDD systems is discussed in this subsection. The proposed LDMA scheme comprises three main stages, including: Initial access, uplink equivalent channel estimation, and uplink/downlink data transmission, which is shown in Fig.~\ref{img: ldma procedure}.
\par In the initial access procedure, BS could perform beam sweeping by virtue of synchronization signal (SS) bursts according to 3rd Generation
Partnership Project (3GPP) specifications to establish a physical link connection to idle UEs~\cite{R15'38.211,R15'doc'17}. The beam sweeping is usually performed based on a predefined codebook~\cite{Giordani'18't}. However, as discussed before, classical far-field DFT-based codebooks are no longer suitable for near-field communications, a near-field codebook design method has to be investigated. After establishing a $M$-codeword near-field codebook $\mathcal{W} = [w_1, \cdots, w_M]$, BS could select the codeword one in a SS block from the predefined codebook periodically to search for the best beam. Then, the UE could report the beam quality and beam decision through random access channel (RACH) resources. Afterward, the best beam could be employed at BS as the uplink combiner or downlink precoder for subsequent channel estimation and data transmissions.
\begin{figure}[!t]
	\centering
	\setlength{\abovecaptionskip}{0.cm}
	\includegraphics[width=3in]{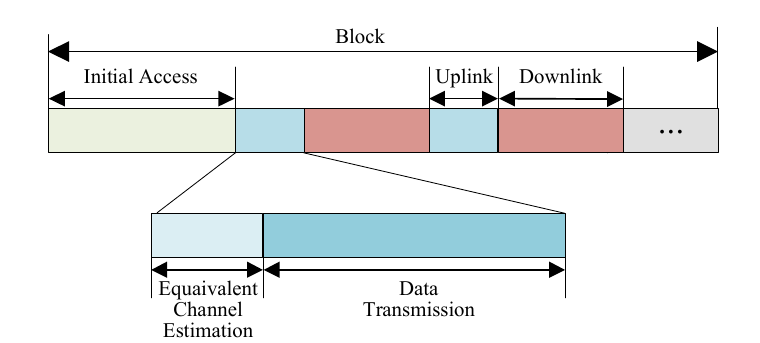}
	\caption{Diagram of the communication procedure for LDMA scheme.}
	\label{img: ldma procedure}
\end{figure}

\par After the classical beam sweeping in the initial access procedure, we assume each served UE is scheduled with a specific codeword as the analog combiner/precoder ${\bf{w}}_k$. Then, the analog precoding for all connected UEs can be designed as ${\bf{F}}_{\rm{A}} = [{\bf{w}}_1, {\bf{w}}_2, \cdots, {\bf{w}}_K]$. Afterward, in the uplink stage, the BS could estimate the effective channel by non-orthogonal pilots sent from all UEs as
\begin{equation}
\label{eq: effective channel}
\begin{aligned}
{\overline{\bf{h}}}_{k} = {\bf{F}}_{\rm{A}}^H {\bf{h}}_k + {\bf{F}}_{\rm{A}}^H {\bf{n}}_k,
\end{aligned}
\end{equation}
where ${\bf{n}}_k$ denotes the noise corresponding to the $k^{th}$ UE.
Finally, the BS could design the digital precoder through the estimated effective channel by weighted minimum mean square error (WMMSE) as in \cite{He'11'j} or zero-forcing (ZF) as
\begin{equation}
\label{eq: zero forcing}
\begin{aligned}
{\bf{F}}_{\rm{D}} = {\overline{\bf{H}}}^H \left({\overline{\bf{H}}}\ {\overline{\bf{H}}}^H\right)^{-1} {\bf{\Lambda}},
\end{aligned}
\end{equation}
where ${\overline{\bf{H}}} = [{\overline{\bf{h}}}_1, \cdots, {\overline{\bf{h}}}_K]^H$ denotes the effective channel of all $K$ UEs. The diagonal matrix ${\bf{\Lambda}}$ denotes the power allocation for different UEs, which is designed to satisfy $\|{\bf{F}}_{\rm{A}} {\bf{f}}_{{\rm{D}},k}\|^2 = 1$. To sum up, the LDMA scheme can be summarized in {\bf{Algorithm}~\ref{alg:2}}.
\begin{algorithm}[!t] 
	\caption{Location Division Multiple Access.} 
	\label{alg:2} 
	\begin{algorithmic}[1]
		\REQUIRE ~ 
		Multi-user channel ${\bf{H}}$.
		\ENSURE ~ 
        Digital precoder ${\bf{F}}_{\rm{D}}$ and analog precoder ${\bf{F}}_{\rm{A}}$ 
		\STATE {\bf{Codebook Design:}}
        \STATE Constructs the near-field codebook $\mathcal{W}$;
        \STATE {\bf{Initial Access:}} 
        \STATE BS performs beam sweeping and each UE performs beam determination and reports the beam decision to BS;
		\STATE BS selects the best codeword ${\bf{w}}_k$ for $k^{\rm th}$ UE from $\mathcal{W}$ and construct the analog precoder ${\bf{F}}_{\rm{A}} = [{\bf{w}}_1, {\bf{w}}_2, \cdots, {\bf{w}}_K]$;
        \STATE {\bf{Uplink Equivalent Channel Estimation:}} 
        \STATE Each UE sends uplink non-orthogonal pilots $x_{k}^{p}$;
        \STATE BS estimates the effective channel for $k^{th}$ UE by~\eqref{eq: effective channel};
        \STATE BS designs digital combiner and precoder ${\bf{F}}_{\rm{D}}$ by~\eqref{eq: zero forcing};
        \STATE {\bf{Uplink/Downlink Data Transmission:}}
		\STATE BS performs hybrid combining/precoding with ${\bf{F}}_{\rm{D}}$ and ${\bf{F}}_{\rm{A}}$;
		\RETURN Digital precoder ${\bf{F}}_{\rm{D}}$ and analog precoder ${\bf{F}}_{\rm{A}}$.
	\end{algorithmic}
\end{algorithm}

\par It is worth noting that, the communication procedure proposed in this paper is not the only way to realize LDMA. Following the idea of acquiring near-field channels and exploiting the communication resources in the extra dimension, other near-field channel estimations~\cite{Cui'22'tcom,Friedlander'19'tsp}, beam sweeping~\cite{You'22'beam} and precoding methods~\cite{zhang'22'j} can also be incorporated into the LDMA scheme.

\section{Performance Analysis}\label{sec: per ana}
In this section, the performance analysis of LDMA is provided. We first show that the asymptotic spectrum efficiency of proposed LDMA could achieve the ideal capacity without multi-user interferences. Then, we explore the performance gain under a special scenario where multiple UEs are located within the same direction, which reveals that LDMA could elaborately utilize the spatial resources to enhance system performance compared with classical far-field SDMA.
\subsection{Asymptotic Capacity of Single-path Channel}\label{sec: asy single}
To investigate the system performance of the proposed LDMA, we first assume the single-path channel for each UE for ULA systems for simplicity. Then, the channel model could be rewritten as
\begin{equation}
\label{eq: single path}
\begin{aligned}
{\bf{h}}_k = \sqrt{N} \alpha_k {\bf{b}}(r_k, \theta_k, \phi_k),
\end{aligned}
\end{equation}
where $\alpha_k$, $r_k$, $\theta_k$ and $\phi_k$ denote the complex path gain, distance between the UE and the array center, and the elevation and azimuth angle of the $k^{\rm th}$ UE, respectively. Here we assume that the BS acquires the perfect channel state information. In addition, the infinite analog codebook is adopted, which means no quantization of codebooks has to be performed and the BS could perfectly select the conjugate transpose of near-field beam focusing vector as the analog precoder. Therefore, the optimal analog precoder for all $K$ UEs can be written as ${\bf{F}}_{\rm{A}} = {\bf{B}} = [{\bf{b}}(r_1, \theta_1, \phi_1), \cdots, {\bf{b}}(r_K, \theta_K, \phi_K)]$. For analysis simplicity, we assume ZF is adopted for designing the digital precoder. In addition, the large-scale fading is neglected and thus different UEs share the same path gain. Then, the spectrum efficiency $R$ can be formulated through the following lemma.

\begin{lemma}
\label{lemma5}
Assume that the single-path channel model is adopted for all $K$ UEs, and the BS acquires perfect channel state information. Then, with equal power allocation, the downlink spectrum efficiency achieved by {\bf{Algorithm}~\ref{alg:2}} is given by 
\begin{equation}
\label{eq: EE single}
\begin{aligned}
R_N = \sum_{k=1}^K R_k = \sum_{k=1}^K \log_2 \left(1+\frac{P}{K\sigma^2} \frac{N |\alpha_k|^2}{[{\bf{B}}^H{\bf{B}}]_{k,k}^{-1}} \right),
\end{aligned}
\end{equation}
where $P$ denotes the total transmission power, $[{{\bf{B}}^{H}}{\bf{B}}]_{k,k}^{-1}$ denotes the $k^{\rm th}$ diagonal entry of the matrix $({{\bf{B}}^{H}}{\bf{B}})^{-1}$, and $\sigma^2$ is the noise power at all the $K$ users. 
\end{lemma}
\begin{proof}
The proof is provided in {\bf Appendix~\ref{app: achievable rate}}.
\end{proof}
For multi-user downlink MIMO systems, an ideal communication scenario is that multiple UEs could communicate with BS without any interference. The ideal capacity could be expressed as
\begin{equation}
\label{eq: capacity opt}
\begin{aligned}
\hat{R} = \sum_{k=1}^K \log_2 \left(1+\frac{P}{K\sigma^2} N|\alpha_k|^2 \right).
\end{aligned}
\end{equation}
Due to the asymptotic orthogonality proved in {\bf Corollary~\ref{coro2}}, the asymptotic spectrum efficiency could be shown to achieve the ideal capacity according to the following corollary.

\begin{corollary}
\label{coro3}
Assume a fixed number of UEs $K$. If the number of BS antennas $N$ tends to infinity, then the spectrum efficiency, as a function of $N$, approaches the ideal capacity with probability one, i.e., 
\begin{equation}
\label{eq: asymptotic spectrum efficiency}
\begin{aligned}
\mathbb{P}\left[ \lim_{N \to +\infty} R = \hat{R}\right] = 1.
\end{aligned}
\end{equation}
\end{corollary}

\begin{proof}
The proof is provided in {\bf Appendix~\ref{app:proof_coro_3}}.
\end{proof}
Therefore, as the number of antennas $N$ tends to infinity, the near-optimal spectrum efficiency performance could be ensured.

\subsection{Linear Distribution Analysis}\label{sec: per B}
Compared with SDMA, LDMA possesses extra resolution in the distance domain, which is potential to distinguish UEs in the same direction. The ability to serve single-path UEs residing in the same angle could be viewed as a feature that distinguishes itself from SDMA. As a result, we will investigate the spectrum efficiency improvement by considering a special scenario where UEs are linearly distributed.

First, we investigate a simple ULA system, in which three UEs are linearly distributed in a certain direction $\theta$. The UEs are located within the distance range of $[r_{\rm{min}}, r_{\rm{max}}]$. Without loss of generality, we assume $r_1 \leq r_2 \leq r_3$. Thus, the channels can be expressed as ${\bf{h}}(r_i, \theta)$, $i=1,2,3$. According to {\bf{Lemma}~\ref{lemma5}}, the analog precoder can be obtained as ${\bf{B}}_{\rm{TU}} = [{\bf{b}}(r_1, \theta),{\bf{b}}(r_2, \theta), {\bf{b}}(r_3, \theta)]$. The maximum spectrum efficiency of the three UEs system can be obtained based on the following assumptions:
\begin{enumerate}[(i)]
\item Single-path channel is adopted for each UE, and high SNR is assumed.
\item The interference of non-adjacent UEs is neglected, i.e. the interference of the $1^{\rm st}$ UE and $3^{\rm rd}$ UE is approximated by 0.
\item The analog precoder is determined by {\bf{Algorithm~\ref{alg:2}}}, and ZF digital precoder is adopted.
\end{enumerate}
Then, we can obtain 
\begin{equation}
\label{eq: T three users}
{\bf{T}} = {\bf{B}}^H {\bf{B}} = 
\left[
\begin{array}{ccc}
1 & \delta_{21}^* & 0 \\
\delta_{21} & 1 & \delta_{32}^* \\
0 & \delta_{32} & 1
\end{array}
\right],
\end{equation}
where $\delta_{ij} = {\bf{b}}^H(r_i, \theta) {\bf{b}}(r_j, \theta)$ denotes the inner product of the $i^{\rm th}$ and $j^{\rm th}$ beam focusing vectors. Suppose  the locations of the first and third UE are fixed, then the maximum achievable spectrum efficiency could be approximated by the following {\bf Lemma~\ref{lemma6}}.
\begin{lemma}
\label{lemma6}
We assuming $\delta_{21} = g(x)$ and $\delta_{32} = g(r_0 - x)$, where $g(x)$ is defined as $g(x): \mathbb{R}_+ \to \mathbb{C}$, and $|g(x)|^2$ is monotonically decreasing and convex, satisfying $0 \leq |g(x)| \leq 1$. Predefined parameter $r_0$ satisfies $r_0>0$. Then the spectrum efficiency defined in~\eqref{eq: EE single} satisfies 
\begin{equation}
\label{eq: three user max SE}
\begin{aligned}
R(x) \lesssim R(x)^{\rm aub} = 2\log_2 \left(1+\frac{P}{K\sigma_n^2} \frac{N |\alpha|^2 (1-2g(\hat{x}))}{(1-g(\hat{x}))} \right) \\
+ \log_2 \left(1+\frac{P}{K\sigma_n^2} N |\alpha|^2 (1-2g(\hat{x})) \right).
\end{aligned}
\end{equation}
The equality holds when $x = \hat{x}$ satisfying $g(\hat{x}) = g(r_0-\hat{x})$.
\end{lemma}
\begin{proof}
The proof is provided in {\bf Appendix~\ref{app:proof_lemma6}}.
\end{proof}
\begin{remark}
It can be proved that the envelope of $G(\beta)$ satisfies the monotonically decreasing and convex constraint. Therefore, the maximum spectrum efficiency can be approximately obtained through {\bf{Lemma}~\ref{lemma6}} for a three-UE system. It is worth noting that, this conclusion is drawn under the assumption that interference from non-adjacent UEs is neglected. Therefore, this lemma can be viewed as an approximated upper bound, i.e., $R\lesssim R^{\rm aub}$ for real communication scenarios.
\end{remark}

Under the three assumptions, the same conclusion of {\bf{Lemma}~\ref{lemma6}} can be generalized into multi-user scenarios. Specifically, we consider a system where multiple UEs are linearly distributed. If we investigate any adjacent three UEs, the middle UE indexed by $i$ must be located in the position that makes the same interference on the two adjacent UEs to maximize the spectrum efficiency, which is to say $|{\bf{b}}_i^H{\bf{b}}_{i+1}| = |{\bf{b}}_i^H{\bf{b}}_{i-1}|$ for $i=2,3,\cdots,K-1$. Otherwise, the system spectrum efficiency can be further enhanced by adjusting the location of the middle UE. Then, when $K$ UEs are placed satisfying this condition, the maximum spectrum efficiency of multiple UEs linearly distributed can be obtained through the following lemma.
\begin{lemma}
\label{lemma7}
The approximated upper bound of the spectrum efficiency of multiple linearly distributed UEs can be expressed as
\begin{equation}
\label{eq: multiple users maximum}
\begin{aligned}
R \lesssim  R^{\rm{aub}} = \sum_{k=1}^K \log_2 \left(1+\frac{P}{K\sigma_n^2} \frac{N |\alpha_k|^2}{\gamma_k} \right),
\end{aligned}
\end{equation}
where $\gamma_k$ is determined by
\begin{equation}
\label{eq: gamma formulation}
\begin{aligned}
\gamma_k = \frac{(\chi_1 x_1^{k-2}+\chi_2x_2^{k-2})(\chi_1x_1^{K-k-1}+\chi_2x_2^{K-k-1})}{\chi_1x_1^{K-1}+\chi_2x_2^{K-1}},
\end{aligned}
\end{equation}
and $x_1, x_2$ are solutions to $x^2-x+|\delta|^2=0$, and $\chi_1 = -\frac{x_1^2}{x_2-x_1}$, $\chi_2 = \frac{x_2^2}{x_2-x_1}$, $|\delta|$ is defined as the min-max correlation between different UEs as $|\delta| = \mathop{\min}\limits_{r_1,\cdots,r_K} \mathop{\max}\limits_{i\neq j} |{\bf{b}}^H(r_i,\theta,\phi){\bf{b}}(r_j,\theta,\phi))|$.
\end{lemma}
\begin{proof}
The proof is provided in {\bf Appendix~\ref{app:proof_lemma7}}.
\end{proof}

The approximated upper bound of the spectrum efficiency will be further verified in the following section.

\section{Simulation Results}\label{sec: sim}
In this section, we evaluate the performance of the proposed LDMA scheme through numerical simulations. A single-cell communication scenario is considered, where multiple single-antenna UEs are served by one BS equipped with an ELAA. The carrier frequency is set to $30$ GHz and the array is half-wavelength spaced. Thus the spacing is $d = \lambda/2 = 0.5\,{\rm cm}$. The 512-element ULA and $256 \times 16$-element UPA are considered in simulations. To reveal the comprehensive performance of the LDMA scheme, two scenarios are considered for both ULA and UPA systems, including the uniformly distributed UEs and linearly distributed UEs.
\subsection{Linear Distribution Scenario for ULA Systems}\label{sec: sim linear}
To validate the conclusion of the approximated upper bound of spectrum efficiency for linearly distributed UEs, we first conduct a simulation where the UEs are aligned along the spatial angle $\phi = 0$. Each UE is located within the range of $[r_{\rm{min}}, r_{\rm{max}}]$, where $r_{\rm{min}} = 4\,{\rm m}$ and $r_{\rm{max}} = 150\,{\rm m}$. The number of UEs varies in $[1, 14]$. SNR is set to be $12\, {\rm{dB}}$. The infinite-size codebook and single-path channel are adopted as introduced in Section~\ref{sec: asy single}. 
\par The spectrum efficiency with increasing number of UEs is plotted in Fig.~\ref{img: linear maximum SE}. The red line and blue line represent the situation where UEs are placed to minimize the interference from adjacent UEs without and with non-adjacent (NA) interference, respectively. The red line denotes the approximated upper bound derived in~\eqref{eq: multiple users maximum} neglecting the NA interference, which is unreachable since NA interference is neglected. A reachable performance for the same user position is plotted in blue. An exhaustive search is performed to search for the reachable maximum spectrum efficiency, which is plotted in green. It shows that the approximated upper bound in~\eqref{eq: multiple users maximum} is tight for small number of UEs. Therefore, {\bf{Lemma}~\ref{lemma6}} provides a relatively accurate estimation of ideal spectrum efficiency.
\par In addition, the orange line denotes the spectrum efficiency for randomly and linearly distributed UEs. The black dashed line denotes the spectrum efficiency employing far-field SDMA. In classical far-field SDMA, only one UE could access to the BS since UEs are linearly distributed and thus the single-path channels are constructed based on the same steering vector. Therefore, since the large-scale fading is omitted, the spectrum efficiency becomes a constant in SDMA scheme. It shows that even randomly distributed UEs with LDMA also outperform far-field SDMA when the number of UEs is not very large.
\begin{figure}[!t]
	\centering
	\setlength{\abovecaptionskip}{0.cm}
	\includegraphics[width=3in]{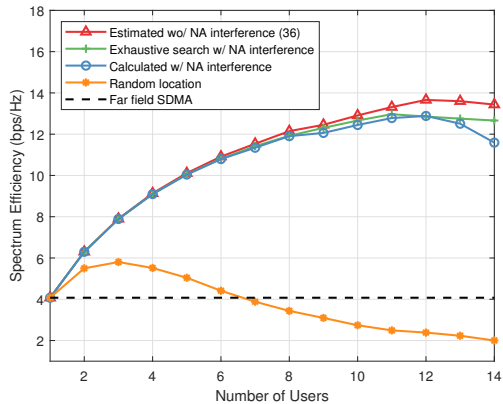}
	\caption{Spectrum effciency achieved under different assumptions.}
	\label{img: linear maximum SE}
    \vspace{-3mm}
\end{figure}
\par We consider a more realistic multi-path model with ULA where both the LoS and NLoS components exist. The UEs are located within the range of $[r_{\rm{min}}, r_{\rm{max}}]$ where $r_{\rm{min}} = 4$ m and $r_{\rm{max}} = 100$ m. The number of UEs is $K=4$ and the number of NLoS paths is $L=5$. The near-field polar-domain codebook introduced in ~\cite{Cui'22'tcom} is adopted for initial access as in {\bf{Algorithm}~\ref{alg:2}}. The baseline is the far-field DFT codebook adopted in~\cite{Heath'15'j} as classical far-field SDMA. Both ZF and WMMSE are considered for designing digital precoders for LDMA and SDMA. The spectrum efficiency against SNR is shown in Fig.~\ref{img: line ULA}. It can be shown that the proposed LDMA outperforms SDMA on both ZF and WMMSE scenarios. The WMMSE-based LDMA and ZF-based LDMA achieve about $60\%$ and $240\%$ performance gain compared with WMMSE- and ZF-based SDMA at SNR $=\,20$ dB, respectively. 
\par In addition, ZF-based fully-digital precoder assuming a perfect channel state information is also considered to reveal the performance upper bound. It can be shown that ZF-based fully-digital precoding scheme outperforms all other schemes in high SNR scenarios. Also, we consider a hybrid precoding method that approximates the fully-digital precoder as in~\cite{Yu'16'j}. The hybrid precoding is slightly worse than the fully-digital scheme, but it is hard to reach since perfect channel estimation is hard to obtain. The WMMSE-based LDMA also outperforms the ZF-based fully-digital SDMA in low SNR scenarios since ZF enlarges the noise. It is worth noting that the high spectrum efficiency in simulations is achieved through multi-stream parallel transmissions with multiple users. The spectrum efficiency for each user is several bps/Hz, and therefore could be supported with current protocols.
\par To summarize, the performance gain of LDMA is regarded to originate from two aspects. First, the BS could harvest the extra orthogonality in the distance domain to accommodate more users without significant interferences, which has been shown in Fig.~\ref{img: linear maximum SE}. Second, the near-field beamforming techniques could form higher beamforming gain compared with far-field beamforming, consequently improving the spectrum efficiency~\cite{Cui'22'tcom}. Both effects are coupled to contribute to the performance gain of LDMA.
\begin{figure}[!t]
	\centering
	\setlength{\abovecaptionskip}{0.cm}
	\includegraphics[width=3in]{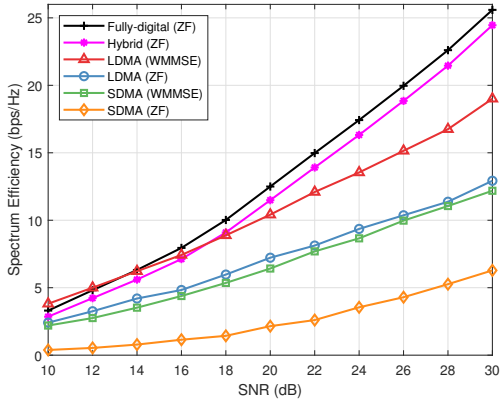}
	\caption{Comparison of the proposed LDMA and classical far-field multiple access under the linear distribution assumption.}
	\label{img: line ULA}
    \vspace{-3mm}
\end{figure}
\subsection{Uniform Distribution Scenario for ULA Systems}\label{sec: sim uniform}
\par Next, we consider a more general scenario where UEs are uniformly distributed in the cell. The UEs are assumed to be uniformly distributed in an area with a radius range $[4\,{\rm m}, 100\,{\rm m}]$ and spatial angle range $[-\pi/3$, $\pi/3]$, which is commonly adopted for one sector. The multi-path channel model is adopted with the number of NLoS paths $L=5$. The number of UEs is set to $K = 10$. As in the linear distribution scenario, both the ZF and WMMSE are adopted for designing the digital precoder. The simulation results are shown in Fig.~\ref{img: uniform ULA}, where the WMMSE-based LDMA achieves nearly $90\%$ improvement in spectrum efficiency at SNR $=\,20$ dB compared with classical WMMSE-based SDMA. Also, since the near-field precoding could be leveraged to manage inter-user interferences, the performance gap between LDMA with far-field SDMA increases as SNR increases. 
\begin{figure}[!t]
	\centering
	\setlength{\abovecaptionskip}{0.cm}
	\includegraphics[width=3in]{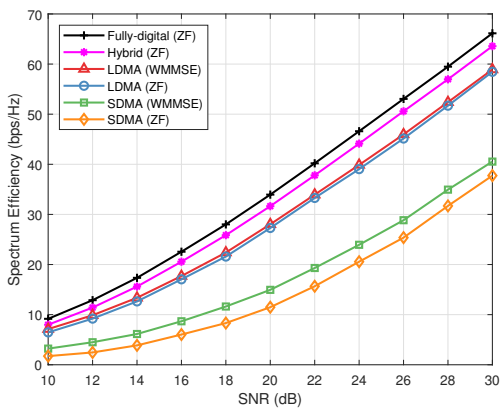}
	\caption{Comparison of the proposed LDMA and classical SDMA under the uniform distribution assumption.}
	\label{img: uniform ULA}
\end{figure}
\par To validate the performance of LDMA with different system parameters, we also conduct simulations adjusting the number of paths, Rician factor, and number of antennas. The spectrum efficiency against number of NLoS paths is plotted in Fig.~\ref{img: uniform ULA L}. It is shown that the performance of LDMA and SDMA both decrease as $L$ increases. The reason lies in that the analog precoding is designed based on the match-filter principle, which prefers the sparse channel model. In addition, LDMA outperforms SDMA for different $L$, which illustrates the robustness of LDMA for different scattering environments.
\begin{figure}[!t]
	\centering
	\setlength{\abovecaptionskip}{0.cm}
	\includegraphics[width=3in]{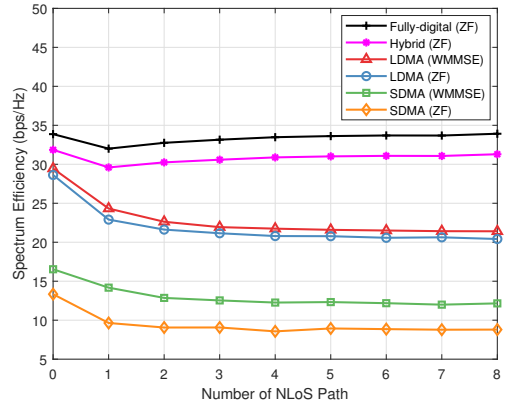}
	\caption{Comparison of the proposed LDMA and classical SDMA for different number of NLoS channels.}
	\label{img: uniform ULA L}
\end{figure}
\par Furthermore, the spectrum efficiency against different $\kappa$ is shown in Fig.~\ref{img: uniform ULA kapa}, where LDMA always outperforms SDMA for different $\kappa$. Since the analog precoder is based on the single-path channel, the spectrum efficiency performance increases as $\kappa$ increases, which results in a stronger LoS path component.
\begin{figure}[!t]
	\centering
	\setlength{\abovecaptionskip}{0.cm}
	\includegraphics[width=3in]{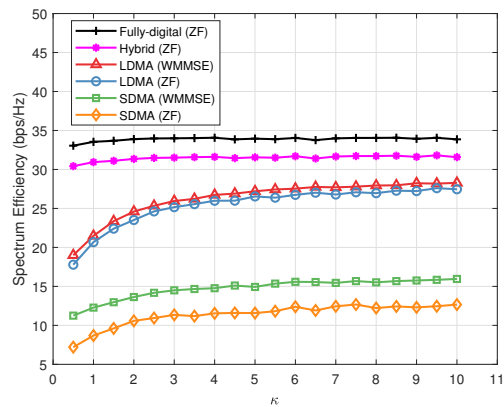}
	\caption{Comparison of the proposed LDMA and classical SDMA for different $\kappa$ of Rician channel.}
	\label{img: uniform ULA kapa}
\end{figure}
\par To show the trend of performance when the channel model changes from far-field to near-field, we perform a simulation where the number of antenna elements changes from $64$ to $512$. The simulation results are shown in Fig.~\ref{img: uniform ULA Nt}. It is shown that ZF- and WMMSE-based far-field SDMA deteriorates as the antenna array is enlarged, which contributes to a stronger near-field effect. The simulations show the necessity of employing LDMA in ELAA systems. In addition, the performance of near-field LDMA obtains a slight increase with the increasing array aperture, which is due to the stronger interference management of near-field focusing beams. 
\begin{figure}[!t]
	\centering
	\setlength{\abovecaptionskip}{0.cm}
	\includegraphics[width=3in]{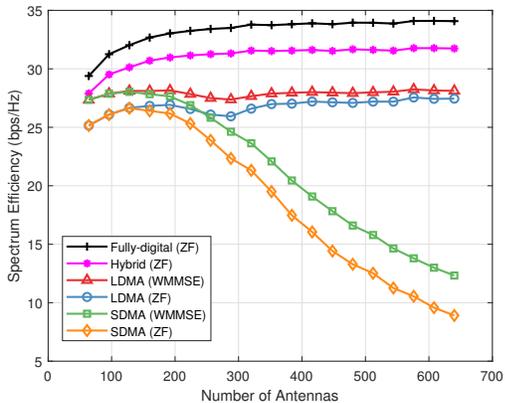}
	\caption{Comparison of the proposed LDMA and classical SDMA for different number of antenna elements.}
	\label{img: uniform ULA Nt}
\end{figure}
\subsection{Performance of UPA Systems}\label{sec: sim upa}
To verify the effectiveness of the proposed codebook and the effectiveness of LDMA in UPA systems, we consider an $N_1 \times N_2 = 256 \times 16$ UPA working at $30$ GHz. According to~\cite{Liu'17'j}, the multi-path components and angular spread in the elevation domain are limited. Therefore, the number of antennas along the elevation axis is less than the azimuth axis. Similar to ULA systems, UEs are assumed to be linearly or uniformly distributed in the space of $\theta \in [\pi/3, 2\pi/3]$, $\phi \in [-\pi/6, \pi/6]$, and $r \in [4\,{\rm m},\,50\,{\rm m}]$, with $L = 5$ and $\kappa = 8$. Threshold $\Delta = 0.55$ is employed to design the spherical-domain codebook. The spectrum efficiency for $4$ UEs is shown in Fig.~\ref{img: UPA}.
\par The simulation results draw a similar conclusion that the proposed LDMA outperforms classical SDMA under linear or uniform UE distribution scenarios. The proposed LDMA could achieve about $150\%$ and $50\%$ performance gain at SNR $=\,20$ dB over ZF-based SDMA under linear and uniform distribution scenarios, respectively. In addition, the uniform sampling method in~\cite{han'20'j} is employed to construct the near-field codebook for comparison. To ensure a fair comparison, the size of the spherical-domain codebook and uniform sampling codebook are kept the same. It also shows the effectiveness of the spherical-domain codebook over other near-field codebooks.

\begin{figure}[!t]
	\centering
	\subfigbottomskip=2pt
	\subfigcapskip=-5pt
	\subfigure[Linear distribution]{
		\label{level.sub.1}
		\includegraphics[width=3in]{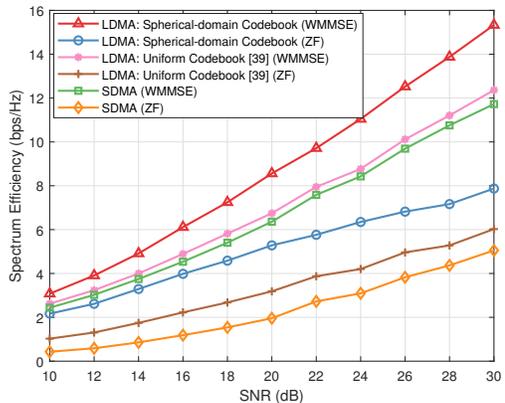}}
	\subfigure[Uniform distribution]{
		\label{level.sub.2}
		\includegraphics[width=3in]{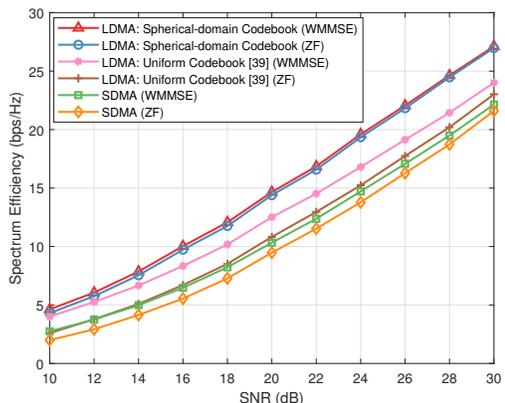}}
	\caption{Comparison of the proposed LDMA and classical SDMA for UPA systems.}
	\label{img: UPA}
    \vspace{-3mm}
\end{figure}

\section{Conclusions}\label{sec: conclusion}
In this paper, the concept of LDMA is first proposed, which leverages the near-field beam focusing property to mitigate interferences and significantly enhance spectrum efficiency. Similar to the asymptotic angular orthogonality of far-field beam steering vectors, the asymptotic orthogonality of near-field beam focusing in the distance domain is investigated. A spherical-domain sampling method is proposed for the UPA codebook design. Based on the near-field codebook, the novel LDMA scheme is investigated. Simulation results verify the superiority of proposed LDMA on spectrum efficiency over uniform and linear distribution scenarios. We hope the LDMA scheme could provide a new possibility for spectrum efficiency enhancement compared with classical SDMA in hybrid precoding. The proposed LDMA can contribute to fulfillments of the 10-fold increase in spectrum efficiency for future 6G ELAA systems. For future research, since the polarization of signals also needs to be carefully tuned to realize constructive superposition at receivers, a more practical precoding design considering spatially varied polarization settings at antennas needs to be investigated~\cite{myers2022near}. In addition, the relationship between the number of users and the increasing rate of the spectrum efficiency curve in large SNR conditions, which reveals the available spatial degrees of freedom in the near-field region, is also an important future research topic~\cite{emil'22}.

\appendices
\section{Proof of Corollary~\ref{coro2}}\label{app: coro2}
By assuming $\eta_1 = kd\left(-\sin\phi_l+\sin\phi_m\right)$ and $\eta_2 = kd^2\left(\frac{\cos^2\phi_l}{2r_l} - \frac{\cos^2\phi_m}{2r_m}\right)$, the correlation of near-field beam focusing vectors focusing on $(r_l,\theta_l)$ and $(r_m,\theta_m)$ is expressed as
\begin{equation}
\label{eq: correlation in 2D domain}
\begin{aligned}
&\left|{\bf{b}}_{\rm{L}}^H(r_l, \phi_l) {\bf{b}}_{\rm{L}}(r_m, \phi_m)\right|  = \frac{1}{N} \left| \sum_{n=-\widetilde{N}}^{\widetilde{N}} e^{jk \left(\psi_{{\rm{L}}, r_l, \phi_l}^{(n)} - \psi_{{\rm{L}}, r_m, \phi_m}^{(n)}\right)} \right|\\
& = \frac{1}{N} \left| \sum_{n=-\widetilde{N}}^{\widetilde{N}} e^{jknd\left(-\sin\phi_l+\sin\phi_m\right) + jkn^2d^2\left(\frac{\cos^2\phi_l}{2r_l} - \frac{\cos^2\phi_m}{2r_m}\right)} \right| \\
& = \frac{1}{N} \left| \sum_{n=-\widetilde{N}}^{\widetilde{N}} e^{j(\eta_1 n + \eta_2 n^2)} \right|.
\end{aligned}
\end{equation}
Then, we adopt the limitation of $N \to +\infty$ as
\begin{equation}
\label{eq: correlation in 2D domain limitation}
\begin{aligned}
&\lim_{N \to +\infty} |{\bf{b}}_{\rm{L}}^H(r_l, \phi_l) {\bf{b}}_{\rm{L}}(r_m, \phi_m)| \\
&= \lim_{N \to +\infty} \frac{1}{N} \left| \sum_{n=-\widetilde{N}}^{\widetilde{N}} e^{j(\eta_1 n + \eta_2 n^2)} \right| \\
& \mathop{\approx}\limits^{(a)} \lim_{N \to +\infty} \frac{1}{N} \left| \int_{-\widetilde{N}}^{\widetilde{N}} e^{j(\eta_1 n + \eta_2 n^2)} {\rm{d}}n  \right|,
\end{aligned}
\end{equation}
where approximation (a) is obtained by replacing summation with integral. By assuming $\eta_2 \neq 0$ and substituting $n + \frac{\eta_1}{2\eta_2}$ with $t$, we can obtain
\begin{equation}
\label{eq: correlation in 2D domain limitation 1}
\begin{aligned}
\lim_{N \to +\infty} |{\bf{b}}_{\rm{L}}^H(r_l, \phi_l) {\bf{b}}_{\rm{L}}(r_m, \phi_m)| = \lim_{N \to +\infty} \frac{1}{N} \left| \int_{-\infty}^{\infty} e^{j\eta_2 t^2} {\rm{d}}t  \right|.
\end{aligned}
\end{equation}
According to the Fresnel function~\cite{Sherman'62'j}, the limitation of the integral can be expressed as
\begin{equation}
\label{eq: correlation in 2D domain limitation limit 1}
\begin{aligned}
\int_{-\infty}^{\infty} e^{j\eta_2 t^2} {\rm{d}}t &= 2 \int_{0}^{\infty} \left(\cos(\eta_2t^2) +j\sin(\eta_2t^2) \right){\rm{d}}t\\
& = \sqrt{\frac{\pi}{2\eta_2}} (1 + j). 
\end{aligned}
\end{equation}
\par Therefore, adopting a similar process as in {\bf{Corollary~\ref{coro1}}}, we can derive the limitation of the correlation of two near-field beam focusing vectors as
\begin{equation}
\label{eq: correlation in 2D domain final 1}
\begin{aligned}
&\lim_{N \to +\infty} \left|{\bf{b}}_{\rm{L}}^H(r_l, \phi_l) {\bf{b}}_{\rm{L}}(r_m, \phi_m)\right|  = \lim_{N \to +\infty} \frac{1}{N} \sqrt{\frac{\pi}{\eta_2}} = 0.
\end{aligned}
\end{equation}

Otherwise, if $\eta_2 = 0$, it means the condition $\frac{\cos^2\phi_l}{2r_l} = \frac{\cos^2\phi_m}{2r_m}$ is satisfied. Recall that we assume $r_l \neq r_m$ or $\phi_l \neq \phi_m$. If $r_l \neq r_m$, then $\phi_l \neq \phi_m$ can also be ensured with $\frac{\cos^2\phi_l}{2r_l} = \frac{\cos^2\phi_m}{2r_m}$. Therefore, $\phi_l \neq \phi_m$ is ensured if $\eta_2 = 0$, where the correlation can be rewritten as
\begin{equation}
\label{eq: correlation in 2D domain final 2}
\begin{aligned}
\lim_{N \to +\infty} &\left|{\bf{b}}_{\rm{L}}^H(r_l, \phi_l) {\bf{b}}_{\rm{L}}(r_m, \phi_m)\right| = \lim_{N \to +\infty} \frac{1}{N} \left| \sum_{n=-\widetilde{N}}^{\widetilde{N}} e^{j\eta_1 n} \right|\\
& \mathop{=}\limits^{(b)} \lim_{N \to +\infty} \frac{1}{N}\left| \Xi_N(kd(-\sin\phi_l+\sin\phi_m)) \right|=0,
\end{aligned}
\end{equation}
where (b) is obtained according to~\eqref{eq: far gain}. Therefore, whether $\eta_2 = 0$ or not, the correlation converges to $0$, which completes the proof.

\section{Proof of Lemma~\ref{lemma2}}\label{app: lemma2}
Following the approximation in~\eqref{eq: near-field UPA distance term}, the correlation of two beam focusing vectors focusing on the same direction defined in~\eqref{eq: near lim 2D UPA} can be rewritten as
\begin{equation}
\label{eq: correlation for UPA}
\begin{aligned}
\quad  & \left|{\bf{b}}_{\rm{P}}^H(r_l, \phi, \theta) {\bf{b}}_{\rm{P}}(r_m, \phi, \theta)\right|  =\\
& \quad \quad \quad \frac{1}{N_1N_2} \left| \sum_{n_1}\sum_{n_2} e^{j(\zeta_1 (n_1-\zeta_2n_2)^2 + \zeta_3 n_2^2)} \right|,
\end{aligned}
\end{equation}
where $\zeta_1$, $\zeta_2$ and $\zeta_3$ are determined by
\begin{equation}
\label{eq: zeta}
\begin{aligned}
\zeta_1 &= kd^2\left(1-\cos^2\theta\right)\left(\frac{1}{2r_l} - \frac{1}{2r_m}\right)\\
\zeta_2 &= \frac{\sin\theta\cos\theta\sin\phi}{1-\cos^2\theta}\\
\zeta_3 &= kd^2\left(1-\sin^2\phi\right)\left(\frac{1}{2r_l} - \frac{1}{2r_m}\right).
\end{aligned}
\end{equation}
Then the limitation of the correlation can be approximated when $N_1 \to +\infty$ as
\begin{equation}
\label{eq: correlation for UPA limitation}
\begin{aligned}
& \lim_{N_1 \to +\infty} \left|{\bf{b}}_{\rm{P}}^H(r_l, \phi, \theta) {\bf{b}}_{\rm{P}}(r_m, \phi, \theta)\right| \\
&  \mathop{\approx}\limits^{(a)} \lim_{N_1 \to +\infty} \frac{1}{N_1N_2} \left| \int_{-\widetilde{N}_1}^{\widetilde{N}_1} \int_{-\widetilde{N}_2}^{\widetilde{N}_2} e^{j(\zeta_1 (n_1-\zeta_2n_2)^2 + \zeta_3 n_2^2)} {\rm{d}}n_1 {\rm{d}}n_2 \right|\\
& \mathop{=}\limits^{(b)} \lim_{N_1 \to +\infty} \frac{1}{N_1N_2} \left| \int_{-\widetilde{N}_2}^{\widetilde{N}_2} \left[\int_{-\widetilde{N}_1-\zeta_2t_2}^{\widetilde{N}_1-\zeta_2t_2} e^{j\zeta_1 t_1^2}{\rm{d}}t_1\right] e^{j\zeta_3 t_2^2}  {\rm{d}}t_2 \right|\\
& \mathop{=}\limits^{(c)} \lim_{N_1 \to +\infty} \frac{1}{N_1N_2} \sqrt{\frac{\pi}{2\zeta_1}} \left|(1+j)\int_{-\widetilde{N}_2}^{\widetilde{N}_2} e^{j\zeta_3 t_2^2}  {\rm{d}}t_2 \right|\\
& \mathop{\leq}\limits^{(d)} \lim_{N_1 \to +\infty} \frac{1}{N_1} \sqrt{\frac{\pi}{\zeta_1}} \left[ \frac{1}{N_2}\int_{-\widetilde{N}_2}^{\widetilde{N}_2} \left|e^{j\zeta_3 t_2^2}\right|  {\rm{d}}t_2\right] = 0,
\end{aligned}
\end{equation}
where approximation (a) is derived by replacing summation with integral. Equality (b) is derived by substituting $n_1-\zeta_2 n_2$ and $n_2$ with $t_1$ and $t_2$, respectively. Owing to the property of Fresnel function that $\int_0^{+\infty} \sin(\frac{\pi}{2}t^2) {\rm{d}}t = \int_0^{+\infty} \cos(\frac{\pi}{2}t^2) {\rm{d}}t = \frac{1}{2}$, it can be proved that $\lim_{N_1 \to +\infty} \int_{-\widetilde{N}_1-\zeta_2t_2}^{\widetilde{N}_1-\zeta_2t_2} e^{j\zeta_1 t_1^2}{\rm{d}}t_1 = \sqrt{\frac{\pi}{2\zeta_1}}(1+j)$. Therefore, equation (c) can be obtained. Approximation (d) is obtained by $|\int x {\rm{d}}x| \leq \int |x| {\rm{d}}x$, which completes the proof.

\section{Proof of Lemma~\ref{lemma4}}\label{app: proof of lemma4}
First, by substituting ${\widetilde{r}}_{l}^{(n_1,n_2)}$ with the approximation in (20), we can obtain
\begin{equation}
\label{eq: p_lemma4 1}
\begin{aligned}
& ~~~\left|{\widetilde{\bf{b}}}^H(r_l, \theta, \phi){\widetilde{\bf{b}}}(r_m, \theta, \phi)\right|\\
&= \left|\sum_{n_1} \sum_{n_2} e^{jk(n_1^2d^2(1-\cos^2\theta) + n_2^2d^2(1-\sin^2\theta\sin^2\phi))(\frac{1}{2r_l}-\frac{1}{2r_m})} \right|,
\end{aligned}
\end{equation}
where $n_1$ and $n_2$ are chosen from the range $[-\widetilde{N_1},\widetilde{N_1}]$ and $[-\widetilde{N_2},\widetilde{N_2}]$, respectively. Then, if we replace the summation with the integral and adopt the approximation $\widetilde{N_i} = \frac{N_i-1}{2} \approx \frac{N_i}{2}$ for $i=1,2$, the complicated summation in~\eqref{eq: p_lemma4 1} could be approximated as
\begin{equation}
\label{eq: p_lemma4 2}
\begin{aligned}
&~~~\left|\sum_{n_1} \sum_{n_2} e^{jk(n_1^2d^2(1-\cos^2\theta) + n_2^2d^2(1-\sin^2\theta\sin^2\phi))(\frac{1}{2r_l}-\frac{1}{2r_m})} \right|\\
& \approx \int_{-N_1/2}^{N_1/2} e^{j\frac{\pi}{2}\sqrt{2n_1^2d^2(1-\cos^2\theta)(\frac{1}{r_l}-\frac{1}{r_m})}^2} {\rm{d}}n_1\\
&~~~ \times \int_{-N_2/2}^{N_2/2} e^{j\frac{\pi}{2}\sqrt{2n_2^2d^2(1-\sin^2\theta\sin^2\phi)(\frac{1}{r_l}-\frac{1}{r_m})}^2} {\rm{d}}n_2\\
& \mathop{=}\limits^{(a)} \int_{0}^{\beta_1} \frac{e^{j\frac{\pi}{2}t_1^2}}{\beta_1}  {\rm{d}}t_1 \times \int_{0}^{\beta_2} \frac{e^{j\frac{\pi}{2}t_2^2}}{\beta_2} {\rm{d}}t_2,
\end{aligned}
\end{equation}
where (a) is derived with $\beta_1 = \sqrt{\frac{N_1^2d^2(1-\cos^2\theta)}{2}(\frac{1}{r_l}-\frac{1}{r_m})}$ and $\beta_2 = \sqrt{\frac{N_2^2d^2(1-\sin^2\theta\sin^2\phi)}{2}(\frac{1}{r_l}-\frac{1}{r_m})}$.
\par Finally, according to the definition of Fresnel functions $C(x) = \int_0^x \cos(\frac{\pi}{2}t^2){\rm{d}}t$ and $S(x) = \int_0^x \sin(\frac{\pi}{2}t^2){\rm{d}}t$, each integration in~\eqref{eq: p_lemma4 2} could be further rewritten as
\begin{equation}
\label{eq: p_lemma4 3}
\begin{aligned}
\int_{0}^{\beta_1} \frac{e^{j\frac{\pi}{2}t_1^2}}{\beta_1} {\rm{d}}t_1 &= \frac{C(\beta_1)+jS(\beta_1)}{\beta_1}\\
& = G(\beta_1),
\end{aligned}
\end{equation}
where $G(\beta) = \frac{C(\beta)+jS(\beta)}{\beta}$. Therefore, we can obtain $\left|{\widetilde{\bf{b}}}^H(r_l, \theta, \phi){\widetilde{\bf{b}}}(r_m, \theta, \phi)\right| \approx |G(\beta_1)G(\beta_2)|$, which completes the proof.

\section{Proof of Lemma~\ref{lemma5}}\label{app: achievable rate}
Owing to the single-path assumptions, the optimal analog precoder is ${\bf{F}}_{\rm{A}} = {\bf{B}} = [{\bf{b}}(r_1, \theta_1, \phi_1), \cdots, {\bf{b}}(r_K, \theta_K, \phi_K)]$. Employing ZF digital precoder, the effetive channel gain is written as 
\begin{equation}
\label{eq: effective channel gain}
\begin{aligned}
|{\bf{h}}_k^H {\bf{F}}_{\rm{A}} {\bf{f}}_{D,l}| &= \left[{\bf{H}} {\bf{F}}_{\rm{A}} {\bf{F}}_{\rm{D}}\right]_{k,l}\\
&= \left[{\overline{\bf{H}}}\ {\overline{\bf{H}}}^H ({\overline{\bf{H}}}\ {\overline{\bf{H}}}^H)^{-1} {\bf{\Lambda}}\right]_{k,l}\\
&= \left[{\bf{\Lambda}}\right]_{k,l}.
\end{aligned}
\end{equation}
Since matrix ${\bf{\Lambda}}$ is diagonal, the inter-user interference could be perfectly eliminated and the channel gain for the $k^{th}$ UE is expressed by ${\bf{\Lambda}}_{k,k}$. ${\bf{\Lambda}}$ is determined by $\|{\bf{F}}_{\rm{A}}{\bf{f}}_{{\rm{D}},k}\|_2^2 = 1$, which can be rewritten as
\begin{equation}
\label{eq: lambda}
\begin{aligned}
\|{\bf{F}}_{\rm{A}}{\bf{f}}_{{\rm{D}},k}\|_2^2 = &{\bf{\Lambda}}_{k,:} ({\overline{\bf{H}}}{\overline{\bf{H}}}^H)^{-1} {\overline{\bf{H}}} {\bf{F}}_{\rm{A}}^H {\bf{F}}_{\rm{A}} {\overline{\bf{H}}}^H ({\overline{\bf{H}}}{\overline{\bf{H}}}^H)^{-1} {\bf{\Lambda}}_{:,k}\\
=& {\bf{\Lambda}}_{k,:} ({\bf{D}}^*)^{-1} ({\bf{B}}^H{\bf{B}}{\bf{B}}^H{\bf{B}})^{-1} {\bf{B}}^H{\bf{B}}\\
 & \times {\bf{B}}^H{\bf{B}}{\bf{B}}^H{\bf{B}}({\bf{B}}^H{\bf{B}}{\bf{B}}^H{\bf{B}})^{-1} {\bf{D}}^{-1} {\bf{\Lambda}}_{:,k} \\
=& \left[{\bf{D}}_{k,k}\right]^{-2} \left[{\bf{\Lambda}}_{k,k}\right]^2 \left[({\bf{B}}^H{\bf{B}})^{-1}\right]_{k,k},
\end{aligned}
\end{equation}
where ${\bf{D}} = \sqrt{N}{\rm{diag}}[\alpha_1,\cdots,\alpha_K]$. To meet the power constraint $\|{\bf{F}}_{\rm{A}}{\bf{f}}_{{\rm{D}},k}\|_2^2=1$, the diagonal entities of ${\bf{\Lambda}}$ could be expressed as
\begin{equation}
\label{eq: lambda2}
\begin{aligned}
{\bf{\Lambda}}_{k,k} = \sqrt{\frac{N}{\left[({\bf{B}}^H{\bf{B}})^{-1}\right]_{k,k}}} |\alpha_k|.
\end{aligned}
\end{equation}
Substituting ${\bf{\Lambda}}$ into~\eqref{eq: spectrum efficiency}, we can obtain~\eqref{eq: EE single}, which completes the proof.

\section{Proof of Corollary~\ref{coro3}}\label{app:proof_coro_3}
From equation~\eqref{eq: EE single}, the spectrum efficiency is determined by the diagonal elements of $({\bf{B}}^H{\bf{B}})^{-1}$. We first show that when the number of antenna elements tends to infinity, ${\bf{B}}^H{\bf{B}} = {\bf{I}}$ with probability one. The off-diagonal $(i,j)$ element of ${\bf{B}}^H{\bf{B}}$ is written as ${\bf{b}}^H(r_i, \theta_i, \phi_j) {\bf{b}}(r_j, \theta_j, \phi_j)$. Assume the random variable $(r, \theta, \phi)$ is uniformly distributed within a 3D sphere with arbitrary radius $r_{\rm{cell}}$. According to {\bf Corollary {\ref{coro2}}}, the correlation of different beam focusing vectors are asymptotically orthogonal, leading to
\begin{equation}
\label{eq: lambda3}
\begin{aligned}
\lim_{N \to \infty} {\bf{b}}^H(r_i, \theta_i, \phi_j) {\bf{b}}(r_j, \theta_j, \phi_j) \mathop{=}\limits^{\rm a.s.} 0, ~~\forall i\neq j.
\end{aligned}
\end{equation}
Consequently, $\lim_{N \to \infty} {\bf{B}}^H{\bf{B}} \mathop{=}\limits^{\rm a.s.} {\bf{I}}$ can be ensured, which results in that $\lim_{N \to \infty} R \mathop{=}\limits^{\rm a.s.} \hat{R}$. This completes the proof. 

\section{Proof of Lemma~\ref{lemma6}}\label{app:proof_lemma6}
With ${\bf{T}}$ defined in~\eqref{eq: T three users}, we omit the variable $x$ and use $g_1$, $g_2$ to represent $|g(x)|^2$ and $|g(r_0-x)|^2$, respectively. Then, the diagonal elements of $[{\bf{B}}^H{\bf{B}}]^{-1}$ can be written as
\begin{equation}
\label{eq: three proof 1}
\begin{aligned}
[{\bf{B}}^H{\bf{B}}]_{1,1}^{-1} &= \frac{1-g_2}{1-g_1-g_2} \\
[{\bf{B}}^H{\bf{B}}]_{2,2}^{-1} &= \frac{1}{1-g_1-g_2}\\
[{\bf{B}}^H{\bf{B}}]_{3,3}^{-1} &= \frac{1-g_1}{1-g_1-g_2}.
\end{aligned}
\end{equation}
Substiting the diagonal elements in~\eqref{eq: three proof 1} into~\eqref{eq: EE single} and defining $\frac{PN}{K\sigma_n^2}$ as $\tau$, the spectrum efficiency could be rewritten as
\begin{equation}
\label{eq: three proof 2}
\begin{aligned}
R  &= \log_2 \left(1+ \tau \frac{1-g_1-g_2}{1-g_2} \right) + \log_2 \left(1+ \tau(1-g_1-g_2) \right)\\
& + \log_2 \left(1+ \tau\frac{1-g_1-g_2}{1-g_1} \right)\\
&\mathop{\approx}\limits^{(a)} \log_2 \left(\frac{(1-g_1-g_2)^3}{(1-g_1)(1-g_2)} \right) + \log_2(\tau^3),
\end{aligned}
\end{equation}
where approximation (a) is obtained based on the high SNR assumption $\log_2(1+x) \approx \log_2(x)$. To search for the maximum spectrum efficiency, we focus on the function inside the logarithmic function $L(x) = \frac{(1-g_1-g_2)^3}{(1-g_1)(1-g_2)}$. The derivatives of $L(x)$ can be written as 
\begin{equation}
\label{eq: three proof 3}
\begin{aligned}
\quad \frac{{\rm{d}}L}{{\rm{d}}x} =& \frac{(g_1+g_2-1)^2}{(g_1-1)^2(g_2-1)^2} \Big[g_1'(g_2-1)(-2g_1+g_2+2)\\
& + g_2'(g_1-1)(g_1-2g_2+2)\Big].
\end{aligned}
\end{equation}
Note that $g_1(x) = g_2(r_0-x)$ holds and $g_1(x)$ is monotonically decreasing. If $\hat{x}$ satisfies $g_1(\hat{x}) = g_2(\hat{x})$, the equation $g_1'(\hat{x}) = -g_2'(\hat{x})$ also holds. Therefore, $\frac{{\rm{d}}L}{{\rm{d}}x}|_{x=\hat{x}} = 0$ can be ensured. Owing to the convexity of $g(x)$, $\hat{x}$ is one local maximum point, which completes the proof.

\section{Proof of Lemma~\ref{lemma7}}\label{app:proof_lemma7}
For multiple linearly distributed UEs, the correlation matrix ${\bf{T}}$ can be written as
\begin{equation}
\label{eq: multiple users proof}
{\bf{T}} = {\bf{B}}^H {\bf{B}} = 
\left[
\begin{matrix}
1 & \delta^* & & & & \\
\delta & 1 & \delta^* & &  &\\
 & \delta & \ddots & \ddots &\\
 & & \ddots & \ddots & \delta^* &\\
 & & & \delta& 1& \\
\end{matrix}
\right],
\end{equation}
where $\delta$ denotes the correlation of any adjacent UEs as $\delta = {\bf{b}}_i^H{\bf{b}}_{i+1}$ for $i \in \{1,2,\cdots, K-1\}$. According to the conclusion of the inversion of a tridiagonal matrix \cite{usmani'1994'}, the $k^{th}$ diagonal element of ${\bf{T}}$ can be written as
\begin{equation}
\label{eq: multiple users proof2}
\begin{aligned}
[{{\bf{T}}}]_{k,k}^{-1} = \gamma_k = \frac{\theta_{k-1}\theta_{K-k}}{\theta_K},
\end{aligned}
\end{equation}
where $\theta_k$ follows $\theta_i = \theta_{i-1}-|\delta|^2\theta_{i-2}$, $\theta_{-1}=0$ and $\theta_{0}=1$. Derived from the recurrence formula, $\theta_i$ has the general solution as $\theta_i = \chi_1 x_1^{i-1} + \chi_2 x_2^{i-1}$ where $\chi_1 = -\frac{x_1^2}{x_2-x_1}$, $\chi_2 = \frac{x_2^2}{x_2-x_1}$. $x_1$ and $x_2$ are solutions to the equation $x^2-x+|\delta|^2=0$. Since $|{\bf{T}}|={\bf{B}}^H {\bf{B}} \geq 0$ should be ensured for any $K$, $|\delta|^2 \leq \frac{1}{2}$ is constrained, which ensures that $x_1$ and $x_2$ are real. Substituting the expression of $\theta_i$, the equation of (\ref{eq: gamma formulation}) could be obtained. Moreover, it can be proved that $[{{\bf{T}}}]_{k,k}^{-1}$ is monotonically increasing as $|\delta|$ increases. Therefore, the spectrum efficiency is maximized when $|\delta|$ is minimized for all $K$ UEs as $|\delta| = \mathop{\min}\limits_{r_1,\cdots,r_K} \mathop{\max}\limits_{i\neq j} |{\bf{b}}^H(r_i,\theta,\phi){\bf{b}}(r_j,\theta,\phi))|$. This completes the proof.

\footnotesize
\balance 
\bibliographystyle{IEEEtran}
\bibliography{IEEEabrv,reference}
\vspace{-1cm}
\begin{IEEEbiography}[{\includegraphics[width=1in,height=1.25in,clip,keepaspectratio]{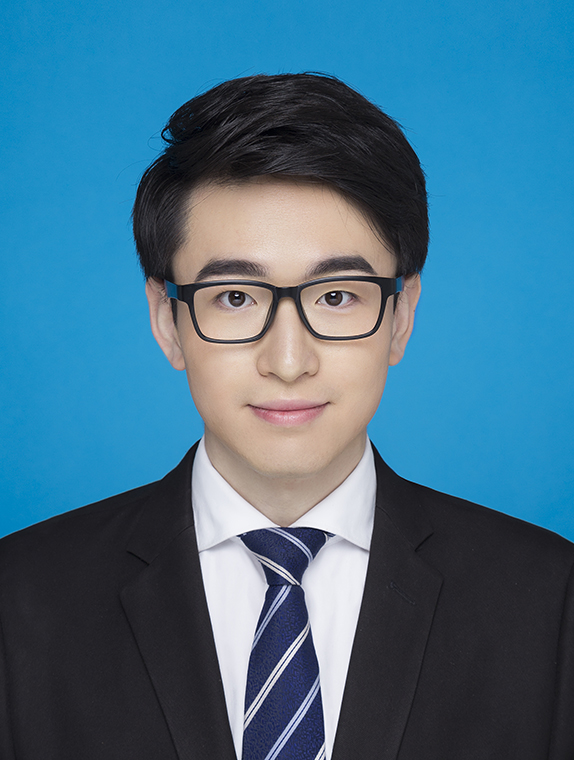}}]{Zidong Wu}
    received the B.E. degree in electronic engineering from Tsinghua University, Beijing, China, in 2019. He is currently pursuing the Ph.D. degree in Department of Electronic Engineering at Tsinghua University, Beijing, China. His research interests include massive MIMO, mmWave communications, machine learning for wireless communications and near-field communications. He has received the Honorary Mention of IEEE ComSoc Student Competition in 2019 and IEEE ICC Oustanding Demo Award in 2022.
\end{IEEEbiography}
\vspace{-1cm}
\begin{IEEEbiography}[{\includegraphics[width=1in,height=1.25in,clip,keepaspectratio]{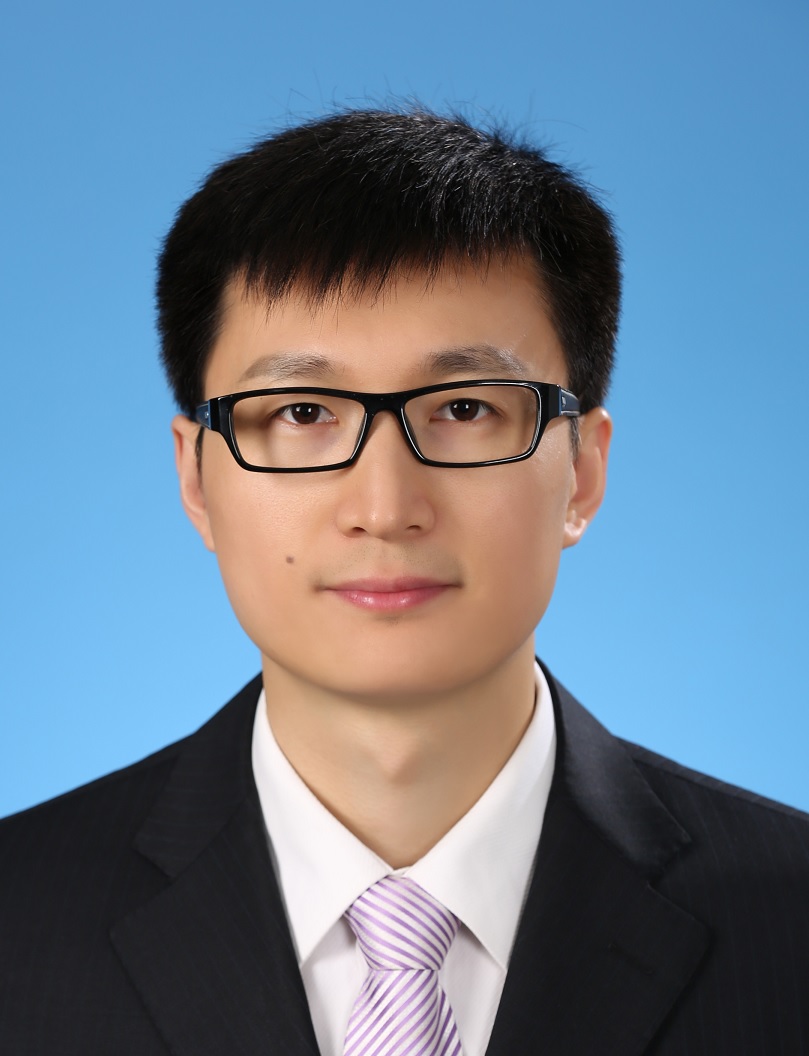}}]{Linglong Dai}
    (Fellow, IEEE) received the B.S. degree from Zhejiang University, Hangzhou, China, in 2003, the M.S. degree from the China Academy of Telecommunications Technology, Beijing, China, in 2006, and the Ph.D. degree from Tsinghua University, Beijing, in 2011. From 2011 to 2013, he was a Post-Doctoral Researcher with the Department of Electronic Engineering, Tsinghua University, where he was an Assistant Professor from 2013 to 2016, an Associate Professor from 2016 to 2022, and has been a Professor since 2022. His current research interests include massive MIMO, reconfigurable intelligent surface (RIS), millimeter-wave and Terahertz communications, wireless AI, and electromagnetic information theory.
    \par He has coauthored the book {\it MmWave Massive MIMO: A Paradigm for 5G} (Academic Press, 2016). He has authored or coauthored over 70 IEEE journal papers and over 40 IEEE conference papers. He also holds 19 granted patents. He has received five IEEE Best Paper Awards at the IEEE ICC 2013, the IEEE ICC 2014, the IEEE ICC 2017, the IEEE VTC 2017-Fall, and the IEEE ICC 2018. He has also received the Tsinghua University Outstanding Ph.D. Graduate Award in 2011, the Beijing Excellent Doctoral Dissertation Award in 2012, the China National Excellent Doctoral Dissertation Nomination Award in 2013, the URSI Young Scientist Award in 2014, the IEEE Transactions on Broadcasting Best Paper Award in 2015, the Electronics Letters Best Paper Award in 2016, the National Natural Science Foundation of China for Outstanding Young Scholars in 2017, the IEEE ComSoc Asia-Pacific Outstanding Young Researcher Award in 2017, the IEEE ComSoc Asia-Pacific Outstanding Paper Award in 2018, the China Communications Best Paper Award in 2019, the IEEE Access Best Multimedia Award in 2020, the IEEE Communications Society Leonard G. Abraham Prize in 2020, the IEEE ComSoc Stephen O. Rice Prize in 2022, and the IEEE ICC Best Demo Award in 2022. He was listed as a Highly Cited Researcher by Clarivate Analytics in 2020-2022. He was elevated as an IEEE Fellow in 2022.
    \par He is currently serving as an Area Editor of the IEEE Communications Letters, and an Editor of the IEEE Transactions on Wireless Communications. He has also served as an Editor of the IEEE Transactions on Communications (2017-2021), an Editor of the IEEE Transactions on Vehicular Technology (2016-2020), and an Editor of the IEEE Communications Letters (2016-2020). He has also served as a Guest Editor of the IEEE Journal on Selected Areas in Communications, IEEE Journal of Selected Topics in Signal Processing, IEEE Wireless Communications, etc. Particularly, he is dedicated to reproducible research and has made a large amount of simulation codes publicly available.
\end{IEEEbiography}

\end{document}